\documentclass[lettersize,journal]{IEEEtran}
\IEEEoverridecommandlockouts
\usepackage[utf8]{inputenc}
\usepackage{color,setspace}
\usepackage{cite}
\usepackage{amssymb,amsfonts}
\usepackage{enumerate}
\usepackage{bbm}
\usepackage{graphicx}
\usepackage{epstopdf}
\usepackage{subfigure}
\usepackage{stfloats}
\usepackage[cmex10]{amsmath}
\usepackage{tasks}
\usepackage{enumitem}
\usepackage{bm}
\usepackage{xcolor}
\usepackage{xspace}
\usepackage{colortbl}
\usepackage{amsthm}
\usepackage{amsmath}
\usepackage{algpseudocode}
\usepackage[ruled,lined,commentsnumbered,linesnumbered]{algorithm2e}
\usepackage{makecell}

\def\BibTeX{{\rm B\kern-.05em{\sc i\kern-.025em b}\kern-.08em
T\kern-.1667em\lower.7ex\hbox{E}\kern-.125emX}}

\newtheorem{theorem}{\textit{Theorem}}
\newtheorem{lemma}{{\textit{Lemma}}}
\newtheorem{proposition}{\textit{Proposition}}
\newtheorem{corollary}{\textit{Corollary}}
\newtheorem{definition}{\textit{Definition}}
\newtheorem{remark}{\textit{Remark}}

\title{An AoI-oriented Time-Frequency Distributed Access Mechanism in Wireless Sensor Networks with Spectrum Division}
\author{\IEEEauthorblockN{Jingwei Liu, Fang Liu, Wing Shing Wong, \textit{Life Fellow}, \textit{IEEE}, Yuan-Hsun Lo, \textit{Member}, \textit{IEEE},\\and Chung Shue Chen, \textit{Senior Member}, \textit{IEEE}}
\thanks{

J. Liu is with the School of Science and Engineering, The Chinese University of Hong Kong (Shenzhen), Shenzhen 518172, China (e-mail: liujingwei@cuhk.edu.cn).

F. Liu is with the College of Electronics and Information Engineering,
Shenzhen University, Shenzhen 518060, China (e-mail: liuf@szu.edu.cn).

W. Wong is with the Department of Information Engineering, The
Chinese University of Hong Kong, Hong Kong SAR, China (e-mail: wswong@ie.cuhk.edu.hk).

Y. Lo is with the Department of Applied Mathematics,
National Pingtung University, Pingtung 900391, Taiwan (e-mail:
yhlo0830@gmail.com).

C. Chen is with Nokia Bell Laboratories, Paris-Saclay Center,
91300 Massy, France (e-mail: chung\_shue.chen@nokia-bell-labs.com).
}
}
\date{August 2026}

\begin{document}





\maketitle

\begin{abstract}
The increasing adoption of spectrum-division techniques enables concurrent uplink transmissions over multiple orthogonal resources, yet low-overhead access design with effective information freshness remains insufficiently studied for large-scale randomly activated sensor networks.
In this paper, we apply the age of information (AoI) to measure information freshness and propose an AoI-efficient deterministic time-frequency distributed access (D-TFDA) mechanism.
D-TFDA combines centralized configuration and distributed operation through a periodic token-based time-frequency structure, which provides sensors with collision-free and predictable transmission opportunities without considerable run-time overhead.
We develop an analytical framework to characterize the long-term average AoI (AAoI) by exploiting the periodicity of the token assignment pattern and modeling the steady local state of each sensor with a one-dimensional discrete-time Markov chain (DTMC).
We further reveal structural properties of the token assignment pattern and identify AoI-equivalent token clusters, which substantially reduce the search space of the AAoI-optimal token allocation problem.
Based on this structure, we formulate the reduced problem as a linear programming (LP) problem and develop an AAoI-optimal search algorithm, together with an auction-inspired heuristic algorithm of lower complexity.
Simulation results validate the proposed AAoI analysis, demonstrate the effectiveness of the token allocation algorithms, and show that D-TFDA achieves substantially lower AAoI than optimized random access baselines by avoiding collisions and exploiting heterogeneous sensor--resource transmission reliability.
\end{abstract}


\section{Introduction}
\IEEEPARstart{I}N recent years, advanced wireless communication standards have incorporated increasingly sophisticated frequency-division multiple access techniques to support dense and heterogeneous wireless connectivity.
Such as physical resource block (PRB) over bandwidth parts in 5G new radio, orthogonal frequency division multiple access (OFDMA) in IEEE 802.11ax/be, and physical uplink shared channel (NPUSCH) in narrowband Internet-of-Things (NB-IoT) \cite{9566742,9442429,9194757}.
By partitioning the available bandwidth into multiple orthogonal frequency-domain resource units (RUs), these systems enable concurrent uplink transmissions from multiple devices over the same time interval, thereby improving spectral utilization, network capacity, and access latency.
As a result, networks built on these standards are anticipated to support a broad range of time-sensitive applications in which the timely delivery of fresh information is crucial.
Typical examples include industrial automation, intelligent transportation, vehicular networks, environmental monitoring, and virtual reality services, where stale observations or control messages may degrade reliability, control accuracy, or user experience \cite{abdullah2015real,10.1145/3695248,s23083880,4550808}.
However, conventional performance metrics, such as throughput, delay, and packet loss probability, cannot fully characterize the freshness of the latest information available at the receiver \cite{sun2022age,kosta2017age,6195689}.
To quantify this notion of information freshness, the age of information (AoI) metric has been introduced, which is defined as the time elapsed since the generation of the most recently received status update at the destination \cite{6195689,sun2022age,5984917,10.1145/3323679.3326520,8514816}.
Therefore, in modern wireless networks with orthogonal frequency-domain resource partitioning, it is essential to design multiple access mechanisms from an AoI perspective rather than relying solely on conventional criteria.

Motivated by this, AoI-oriented multiple access design has received increasing attention in wireless networks with orthogonal transmission resources.
Extensive efforts have investigated centralized AoI-aware scheduling policies, where a central controller coordinates the transmissions of multiple users according to the network state (e.g., see \cite{9055353,9155420,10722850,9792409,9882370}), where the network-wide time-average AoI is typically optimized by jointly determining which users should transmit and which orthogonal resources should be assigned to them over time.
Such strategies are appropriate for scenarios with highly active users and non-large-scale networks \cite{9097306,8954939}.
Alternatively, among the various time-critical applications, the frequency-division techniques are also applicable to wireless dense status-updating sensor networks, where a large number of sensors periodically or randomly report the latest status of dynamic physical processes, such as temperature, motion, or channel conditions, to a central monitor \cite{YICK20082292,6532482}.
In many such networks, sensors enter active periods only sporadically and generate status-update packets intermittently while remaining idle otherwise, resulting in low-activity yet bursty status update traffic in the long term.
This type of behavior commonly arises in, for instance, vehicle networks, where sensors deployed along streets are triggered by the random passage of vehicles, and environmental monitoring networks, where sensors are activated and detect information regularly when wild animals occasionally appear within their sensing range \cite{8340257,10.1145/1869983.1869997}.
For these systems, maintaining fresh status information at the monitor is essential, since outdated messages may degrade monitoring accuracy, inference reliability, or control performance.
However, directly applying centralized scheduling policies to such dense status-updating sensor networks may incur substantial coordination overhead due to local state acquisition and device identification, significantly reducing the efficiency of freshness-related communication.

Accordingly, AoI-oriented distributed access has emerged as a natural design direction for dense status-updating networks, as it can avoid frequent exchange of device-specific scheduling commands and instantaneous local-state reports during network operation.
Under this paradigm, devices determine their transmission attempts according to locally available information and pre-specified access rules.
As a result, unlike centralized scheduling policies that require coordination signaling when arranging transmissions, distributed access can operate with limited run-time signaling and thus maintain low network overhead.
Existing works have investigated random access, which represents a major class of distributed access mechanisms.
The authors of \cite{9785624} proposed an age-optimized slotted-ALOHA protocol for freshness-aware random access.
In \cite{10323421}, Fresh-CSMA was devised as an AoI-oriented carrier-sensing access mechanism for single-hop wireless networks.
In addition to shared-channel random access, the AoI performance of the uplink OFDMA-based random access (UORA) mechanism introduced in IEEE 802.11ax was analyzed and optimized in \cite{10945427}.
Nevertheless, random access remains inherently contention-based, and the AoI performance can be impaired by the uncertainty of successful transmission opportunities.
When multiple active sensors generate status updates within overlapping active periods, bursty transmission demands may intensify contention over the shared transmission resources.
The resulting transmission collisions and backoff procedures in some mechanisms may further introduce irregular service intervals and delay the delivery of fresh updates.
Consequently, although random access is attractive for the low network overhead, its collision-prone and contention-induced randomness may undermine its capability to provide consistently effective AoI performance in the considered networks.
To our knowledge, beyond random access, the design of AoI-efficient distributed access mechanisms in the networks applying orthogonal frequency-domain resources has not been thoroughly investigated in the literature.

{To address this gap, this paper focuses on a wireless sensor network with orthogonal frequency-domain transmission resources.
We consider a time-slotted uplink network, where multiple time-sensitive sensors randomly enter active durations, periodically generate status update packets during each active period, and transmit their latest updates to an access point (AP).
The objective is to develop a low-overhead distributed access mechanism and optimize its long-term average AoI (AAoI) performance under such random status update traffic.}
The main contributions of this work are summarized as follows.
\begin{itemize}
    \item 
    {We develop a deterministic time-frequency distributed access (D-TFDA) mechanism, by which a token assignment is configured before network operation and then used by sensors for independent uplink transmissions.
    D-TFDA inherits the collision-free and predictable service structure of centralized scheduling while preserving the low-overhead operation of distributed access.
    To assess the AoI performance of the D-TFDA network, we capture the transmission opportunities of each sensor over a circular token frame by analyzing the periodicity of the network token assignment pattern.
    Furthermore, we model the steady local state of each sensor using a one-dimensional discrete-time Markov chain (DTMC), which enables tractable characterization of the service time of the status updates of the sensor.
    By doing so, we derive an analytical expression for approximating AAoI of the considered network.}
    \item We formulate the AAoI-optimal token allocation problem, where distinct transmission tokens are assigned to sensors to minimize the network-wide AAoI.
    To efficiently solve this combinatorial problem, we analyze the D-TFDA transmission pattern and reveal the relationship between token indices and occupied transmission resources.
    This leads to the identification of equivalent token clusters (ETCs), where tokens in the same cluster induce identical AAoI performance for the same sensor, thereby substantially reducing the feasible solution space.
    Leveraging the ETC structure, we develop two token allocation algorithms: an LP-based optimal search whose relaxation preserves optimality, and an auction-inspired heuristic that reduces computational and memory overhead while maintaining effective AoI performance in large-scale networks.
    \item We conduct numerical simulations to validate the proposed analytical framework and evaluate the effectiveness of the token allocation algorithms.
    The results show that the derived AAoI expression closely matches the simulated performance, and verify the AoI-equivalent structure revealed by the token pattern analysis.
    Moreover, the proposed LP-based and heuristic search algorithms significantly improve the AAoI performance over the unoptimized token allocation baseline.
    Further comparisons with random access baselines demonstrate that D-TFDA achieves substantially lower AAoI.
\end{itemize}

{The rest of this paper is organized as follows.
Section II reviews the related literature.
Section III introduces the system model and the performance metric.
Section IV presents the proposed D-TFDA mechanism and develops the corresponding AAoI evaluation framework.
Section V formulates the AAoI-oriented token allocation problem and proposes efficient token allocation algorithms.
Section VI offers numerical results to validate the theoretical analysis and evaluate the proposed algorithms.
Finally, Section VII concludes this work.}

\section{Related Work}
{This section introduces existing studies relevant to this work, broadly classified into three categories: random access over a single channel, that over multiple orthogonal frequency-domain resources, and cyclic scheduling.}

AoI-oriented design and analysis in uplink random access networks have attracted substantial research attention, with extensive studies in\cite{9162973,9377549,9785624,10138556,9007478,10323421}.
Specifically, the authors of \cite{9162973,9377549,9785624} investigated and optimized the AoI performance of threshold-based ALOHA networks, where each user attempts transmission with a prescribed probability once its instantaneous AoI exceeds a preset threshold.
In addition, \cite{10138556} presented an analytical characterization of the AoI performance for reservation-assisted framed slotted ALOHA, which consists of both reservation and data slots, and further developed the corresponding optimization framework.
The exploration of AoI for CSMA, another typical random access protocol, has been examined in \cite{9007478,10323421}.
In \cite{9007478}, the AoI behavior of CSMA networks with randomly distributed system parameters was studied. 
By leveraging the stochastic hybrid systems (SHS) framework, the authors obtained a closed-form expression for the network-wide average AoI, followed by further parameter optimization.
V. Tripathi et al. \cite{10323421} developed a distributed AoI-aware access scheme, termed Fresh-CSMA, by exploiting the backoff principle inherent in CSMA for single-hop wireless networks. It was shown that, under the same network state, Fresh-CSMA can reproduce the scheduling actions of the near-AoI-optimal centralized max-weight policy with high probability, and achieves performance close to it. 
Nevertheless, these random access mechanisms were developed for uplink networks transmitting over a single common channel.
Consequently, they cannot be directly applied to the network considered in our work, where the uplink spectrum is partitioned into multiple orthogonal resources.

{Different from the schemes in the above work, the UORA mechanism allows users to contend for transmissions through multiple RUs.
The performance of this mechanism has been studied in \cite{11251194,10945427}.}
Specifically, the authors of \cite{11251194} developed a fixed-point analytical framework for the IEEE 802.11ax hybrid access protocol, jointly considering UORA and scheduled access over RUs.
Based on this and Markov decision process theory, they further devised an optimal dynamic RU allocation policy that balances saturation throughput against average access delay.
However, this work focuses on conventional metrics, rather than optimizing the information freshness.
In \cite{10945427}, the AAoI of UORA networks was systematically analyzed and optimized under stochastic status update arrivals.
The authors constructed two coupled DTMCs to derive an analytical expression for the AAoI.
They further analyzed an approximated AAoI lower bound and developed efficient UORA parameter optimization algorithms that achieve near-optimal freshness performance.
Despite this, UORA is not particularly designed to prioritize freshness, thus does not necessarily provide inherent advantages in terms of~AoI.

Cyclic scheduling constitutes another line of research (e.g., see \cite{10228973,10480158,10189865,10912734} and references therein) related to our work, as it relies on a predesigned periodic transmission pattern to coordinate status update transmission.
In particular, C. Li et al. \cite{10228973} proposed Eywa, a general framework for constructing cyclic schedulers for a family of AoI-related optimization problems in IoT data collection networks with spectrum-division.
The framework was applied to minimize the weighted sum AoI and the required bandwidth under AoI constraints, yielding improved theoretical performance guarantees over existing approaches.
However, Eywa assumes each source has a fixed packet loss probability that is independent of the transmission resource.
This does not capture practical spectrum-division systems, where the transmission reliability of a given source may vary substantially across heterogeneous frequency-domain resources.
In \cite{10912734}, the AoI of massive-scale status-update systems was explored under periodically repeated transmission patterns, while accounting for heterogeneous service times and packet error probabilities.
The authors evaluated the weighted AoI associated with a given cyclic pattern, and further proposed scalable pattern-design algorithms for general networks by combining convex optimization with packet-spreading techniques.
Nevertheless, the underlying network model is restricted to a single shared uplink channel.
Furthermore, these studies adopt a generate-at-will traffic model in which a fresh status update is available whenever a source is scheduled.
Their analytical frameworks therefore cannot be directly extended to the randomly activated sensor network considered in this work. 

{Overall, prior studies have advanced the AoI analysis and design of random access and cyclic scheduling mechanisms.
Nevertheless, their approaches are not readily employed in the scenario considered in our work, which features multiple orthogonal resources and resource-dependent error rates.}

\section{System Model and Preliminaries}
\subsection{Network Model}
We study a time-slotted wireless uplink sensor network consisting of $N$ time-sensitive sensors, indexed by $i\in\{1,2,\cdots,N\}$, and one access point (AP). 
The time slots are indexed by $t\in\{1,2,\cdots,T\}$, where $T$ denotes the time horizon of the network.
The channel bandwidth is divided into $R$ $(R<N)$ equal-sized orthogonal resource units (RUs), indexed by $r\in\{1,2,\cdots,R\}$. 
Each sensor aims to transmit its locally generated status updates to the AP independently and timely, and each transmission is supposed to take one time slot through one RU.
We adopt the error-prone collision channel model, where any RU selected by more than one sensor results in a collision, leading to the failure of all transmissions on that RU.
Additionally, if RU $r$ is solely occupied by the transmission of sensor $i$, the success rate of this transmission is $p_{i,r}\in(0,1]$.
We assume a single-buffer configuration, that is, each sensor maintains a single buffer that only stores the latest status update.
Furthermore, sensors are classified as active and idle sensors.
At the beginning of each time slot, if sensor $i$ is idle, it could be activated and turn into an active sensor with active rate $\lambda_i$.
Once sensor $i$ is activated, it remains in the active state for $M_i$ consecutive slots, which is referred to as one active duration.
In the first time slot of each active duration, sensor $i$ generates a status update packet, and subsequently generates new packets periodically at intervals of $U_i$ slots for the remainder of the active duration.
We illustrate an example of the active and packet generation pattern of sensor $i$ with $M_i=5$ and $U_i=2$.

\begin{figure}[t]
	\centering
	\includegraphics[width=0.48\textwidth]{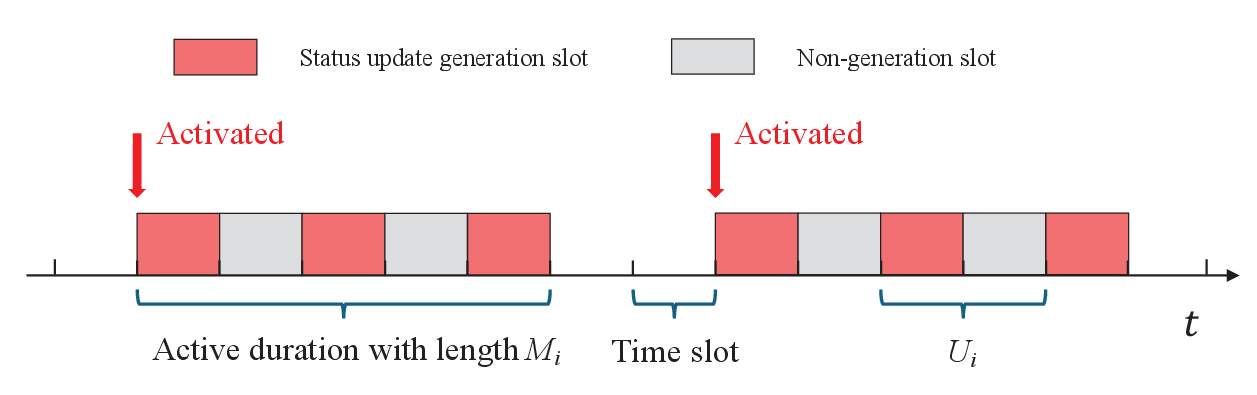}
	\caption{The active and packet generation pattern of sensor $i$ with $M_i=5$ and $U_i=2$.}
	\label{Active}
\vspace{-1em} 
\end{figure}

\subsection{Information Freshness Metric}
We apply the AoI metric, initially proposed in \cite{6195689}, to capture the information freshness across all sensors at the AP.
To mathematically characterize the AoI of the sensors, we first denote the local age, representing the system time of the most recent status update stored in the buffer of sensor $i$ in slot $t$, by $\delta_i(t)$.
The evolution of $\delta_i(t)$ can be presented as
\begin{equation}\label{LAoIevolve}
    \delta_i(t+1)=
    \begin{cases}
        0,& \text{if a status update is generated at}\\
        & \text{sensor $i$ in slot $t+1$,}\\
        \delta_i(t)+1,& \text{otherwise.}
    \end{cases}
\end{equation}
The AP records the current local age of sensor $i$ attached to the content of the status update packet just received from sensor $i$.
Therefore, the evolution of the AoI of sensor $i$ in slot $t$, denoted by $\Delta_i(t)$, can be expressed as
\begin{equation}\label{AoIevolve}
    \Delta_i(t+1)=
    \begin{cases}
        \delta_i(t)+1,& \text{if a status update of sensor $i$ is}\\
        & \text{received by the AP in slot $t$,}\\
        \Delta_i(t)+1,& \text{otherwise.}
    \end{cases}
\end{equation}
In this paper, we adopt the time average expected AoI (AAoI) as the performance metric of the network.
The definition of the network-wide AAoI is given by
\begin{equation}\label{AAoI}
    \overline{\Delta}\triangleq \lim_{T\to\infty}\frac{1}{NT}\mathbb{E}\left[\sum^T_{t=1}\sum^N_{i=1}\Delta_i(t)\right]=\frac{1}{N}\sum^N_{i=1}\overline{\Delta}_i,
\end{equation}
where $\overline{\Delta}_i\triangleq\lim_{T\to\infty}\frac{1}{T}\mathbb{E}\left[\sum^T_{t=1}\Delta_i(t)\right]$ denotes the AAoI of sensor $i$.

The AAoI performance of the considered network can be optimized through applying the centralized multiple scheduled access scheme, where all transmissions are coordinated by the AP.
However, extra overhead is required to schedule the transmission sequences of sensors and to acquire the local buffer status of the sensors under such strategies.
This is inefficient for the considered sensor networks, especially when the status update packets are typically relatively short.
In this context, we are motivated to develop a multiple access mechanism that does not require considerable overhead while achieving effective AoI performance.

\section{D-TFDA Mechanism}
{In this section, we first develop a multiple access mechanism named deterministic time-frequency distributed access (D-TFDA).}
Subsequently, we analyze and evaluate the AAoI performance of the considered network employing the D-TFDA scheme.

\subsection{Mechanism Design}
{The proposed D-TFDA is a centralized-configuration and distributed-operation multiple access mechanism, where the AP preconfigures and broadcasts a deterministic time-frequency access structure before network operation, after which sensors independently transmit over the assigned resources without requiring per-slot scheduling.}
Specifically, we introduce the following key concept.  
\begin{definition}
    We define each RU in each time slot as a unique transmission block (TB).
    In particular, the TB located at RU $r$ in slot $t$ is denoted by $\langle t;r\rangle$.
    Then, we define the order of all TBs.
    Given a slot $t$, we have $\langle t;r_1\rangle < \langle t;r_2\rangle$ for any $r_1<r_2$; given any $r_1,r_2$, we have $\langle t_1;r_1\rangle <\langle t_2;r_2\rangle$ for any $t_1<t_2$.
\end{definition}

The D-TFDA establishes a periodic token-based access structure that leverages all TBs to provide authorizations for uplink transmission access of all sensors in the network.
The implementation of the D-TFDA mechanism is given as follows.
\begin{enumerate}
    \item \underline{\textit{Token-to-TB assignment}}: The AP first specifies the index configuration of the $N$ sensors in the D-TFDA network and the $R$ RUs allocated to the network. 
    Each specific $N,R$ combination corresponds to a unique token assignment over the TBs.
    In particular, given $N$, the D-TFDA provisions $N$ types of tokens, indexed by $\tau\in\{1,2,\cdots,N\}$, which are cyclically assigned to the consecutive TB sequence in ascending order. 
    \item \underline{\textit{Token-to-sensor allocation}}: Subsequently, the AP determines a token allocation policy $\bm{\pi}\triangleq [\pi_1,\pi_2,\cdots,\pi_N]$, where we have $\pi_i=\tau$ denotes that token $\tau$ is allocated to sensor $i$.
    Notice that $\pi_i\neq \pi_{i^{'}}$ for $i\neq i^{'}$.
    \item \underline{\textit{Configuration signaling}}: The AP broadcasts the index information of the sensors, the RUs, and $\bm{\pi}$ to all sensors, thereby each sensor achieves the knowledge of the values of $N,R$ and the corresponding token assignment.
    \item \underline{\textit{Distributed access process}}: During the network operation, each sensor transmits its status update to the AP via the TBs assigned with the token allocated to the sensor in the corresponding time slots and RUs if the buffer of the sensor is not empty.
     
\end{enumerate}

\begin{figure}[t]
	\centering
	\includegraphics[width=0.48\textwidth]{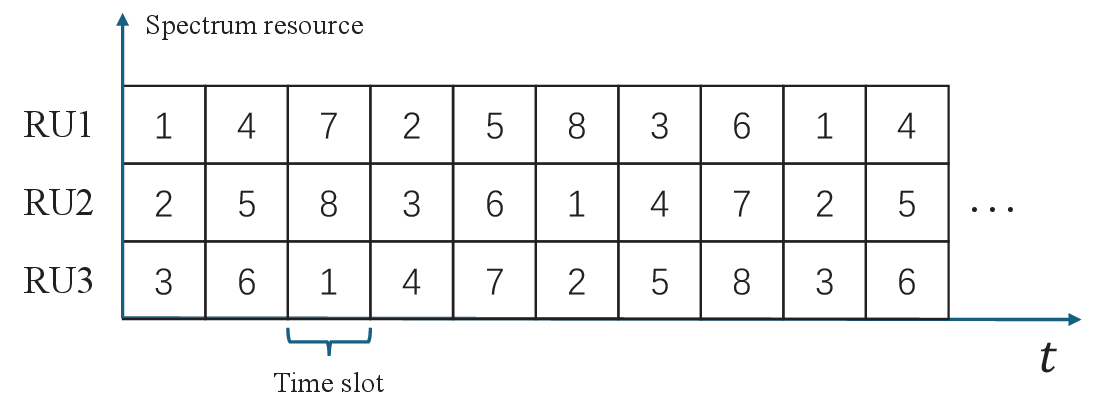}
	\caption{The network token assignment pattern of a D-TFDA network with $R=3$ and $N=8$.}
	\label{TA}
\vspace{-1em}
\end{figure}
Throughout this paper, we distinguish between the terms ``assignment'' and ``allocation'': specifically, ``assignment'' refers to the mapping of tokens to TBs, while ``allocation'' denotes the mapping of tokens to sensors.
The uplink transmissions of the D-TFDA mechanism are coordinated by the token assignment pattern over the TBs.
In this paper, we refer to the assignment pattern of all $N$ types of tokens as the network token assignment pattern and that of a certain token $\tau$ as the single token assignment pattern of token $\tau$.
Fig. \ref{TA} shows an example of the network token assignment pattern with $N=8,R=3$, where each grid represents a TB.
Different single token assignment patterns are orthogonal to each other across the TBs, thereby the transmissions of the different sensors cannot collide.
Clearly, there are $N$ TBs between two consecutive assignments of the same token and $R$ TBs in one time~slot.
Hence, due to $R<N$, each sensor performs at most one transmission in each time slot.

This collision-free, deterministic structure delivers AoI advantages comparable to centralized scheduling, yet without substantial network overhead. 
Moreover, compared to random access schemes, which incur collisions and irregular service intervals, the deterministic token-based structure of D-TFDA ensures each sensor periodic and guaranteed transmission opportunities within $N$ TBs. 
Since the value of AoI is sensitive to the regularity of successful updates, this predictable service structure has the potential to mitigate cumulative AoI. 
We next conduct AAoI performance analysis of the D-TFDA networks.

\subsection{Performance Analysis}
The transmissions of different sensors in the D-TFDA network cannot influence each other, indicating that the state evolutions of the sensors are independent.
The D-TFDA ensures that transmissions of different sensors do not interfere, allowing us to decouple and analyze the AAoI of each sensor independently.
However, the state evolution of a single sensor remains complex.
Specifically, the deterministic token assignment pattern contrasts with the random status generation process, which combines sporadic activation and periodic generation within active periods, creating coupled state transitions involving local age, AoI, and transmission history profile of this sensor. 
A complete characterization would require a multi-dimensional Discrete-Time Markov Chain (DTMC) with a considerably large state space and complex transitions, rendering exact analysis intractable. 
To tackle this challenge, we are driven to evaluate the AAoI performance of a sensor by referring to a universal AAoI expression of a single node in the time-slotted system, presented as \cite{8648195}
\begin{equation}\label{AAoIEq}
    \overline{\Delta}_i=\lim_{T\to\infty}\frac{1}{T}\mathbb{E}\left[\sum^{K_i(T)}_{k=1}Q_i(k)\right],
\end{equation}
where $Q_i(k)$ denotes the cumulative instantaneous AoI values of sensor $i$ during the interval of the $k$th and $(k+1)$th successful transmissions of sensor $i$, and $K_i(T)$ denotes the number of successful transmissions of sensor $i$ by the end of slot $T$.
By \cite{10138556}, we have
\begin{equation}\label{QE}
    \begin{split}
        Q_i(k) & = \sum^{X_i(k)-1}_{z=0} (S_i(k-1)+z)\\
        &= \frac{X_i^2(k)-X_i(k)}{2}+S_i(k-1)X_i(k),
    \end{split}
\end{equation}
where $X_i(k)$ denotes the interval of the $k$th and $(k+1)$th successful transmissions of sensor $i$, and $S_i(k)$ is the service time of the latest generated status update of sensor $i$, i.e., $\delta_i+1$, before the $k$th successful transmission. 
By the renewal theorem \cite{ALSMEYER199437}, we can evaluate $\overline{\Delta}_i$ through deriving expectations of terms in \eqref{QE}.
The tractability of this approach stems from its decomposition structure, which transforms the intractable joint state characterization into expectations over low-dimensional random variables.
Nevertheless, the analyses of these variables are still non-trivial due to that $S_i(k-1)$ and $X_i(k)$ admit dependent statistical properties.
Therefore, rather than evaluating $\overline{\Delta}_i$ by directly applying this approach, we prefer to leverage the specific structural characteristics of the D-TFDA mechanism to permit further simplifications to our AAoI exploration.

To proceed, we first focus on the time slot intervals of consecutive assignments of the token allocated to a sensor, and have the following lemma.
\begin{lemma}\label{Period}
    Given $N$ and $R$, the network token assignment pattern of the D-TFDA mechanism exhibits periodicity over time with the period $\frac{\operatorname{lcm}(N,R)}{R}$ time slots.
\end{lemma}
\begin{proof}
    See Appendix \ref{Proof_Period}.
\end{proof}
We then define the circular token frame (CTF) of token $\tau$ as the token assignments in $\frac{\operatorname{lcm}(N,R)}{R}$ consecutive time slots after a time slot when an assignment of token $\tau$ occurs.
For the example shown in Fig. \ref{TA}, the token assignments from slot $2$ to slot $9$ is the CTF of tokens $1,2,3$, and the token assignments from slot $3$ to slot $10$ is the CTF of tokens $4,5,6$.
By Lemma \ref{Period}, the network token assignment pattern is periodic in the time dimension for the CTF of any token.
In addition, we define the inter-token time (ITT) as the time interval between two consecutive assignments of the token allocated to sensor $i$. 

Since the single token assignment patterns of different types of tokens are the shifted versions along TBs of each other, the AAoI analysis of one sensor can be expanded to that of the other sensors.
As such, we designated sensor $i$ as the typical sensor, allocated with token $\tau$, and focus on the evaluation of $\overline{\Delta}_i$ of this sensor.
Clearly, the corresponding set of single token assignments, i.e., transmission opportunities, of a sensor within this CTF cycles in a period of $\frac{\operatorname{lcm}(N,R)}{R}$ time slots as well.
Moreover, the CTF of each type of token consists of multiple complete ITTs of this token.
Built upon this and inspired by \eqref{AAoIEq} and \eqref{QE}, we tend to analyze $\overline{\Delta}_i$ by evaluating the expected average AoI over the CTF associated with sensor $i$. 
However, the method of computing cumulative instantaneous AoI values as in \eqref{AAoIEq} and \eqref{QE} cannot be directly employed within a CTF, since the intervals corresponding to $Q_i(k)$ that overlap with a CTF may be incomplete.
To address this, we assume that the status update transmitted to the AP successfully will not be discarded from the buffer storing this status update until the next status update preemption.
Under this assumption, the evolutions of $\Delta_i(t)$ and $\delta_i(t)$ remain unaffected, while sensor $i$ always transmits in the last time slot before each ITT.
Furthermore, to simplify the analysis, we adopt an additional critical assumption under which the evolutions of $\Delta_i(t)$ and $\delta_i(t)$ are decoupled from the single token assignment pattern and have reached steady state at the beginning of the CTF of token $\tau$.
Then, we approximate $\overline{\Delta}_i$ following 
\begin{equation}\label{AoICTF}
    \overline{\Delta}_i=\frac{R}{\operatorname{lcm}(N,R)}\mathbb{E}\left[\sum^K_{k=1}\hat{Q}_i(k)\right],
\end{equation}
where $\hat{Q}_i(k)$ represents the cumulative instantaneous AoI values for the $k$th ITT in the CTF of the token allocated to sensor $i$ and $K$ is the number of ITTs in this CTF.
Similar to \eqref{QE}, we present
\begin{equation}\label{QEC}
    \hat{Q}_i(k)=\frac{\hat{X}^2_i(k)-\hat{X}_i(k)}{2}+\hat{D}_i(k)\hat{X}_i(k),
\end{equation}
where $\hat{X}_i(k)$ denotes the duration of the $k$th ITT in the CTF and $\hat{D}_i(k)$ denotes the instantaneous AoI value of sensor $i$ in the first time slot of the $k$th ITT.
To this end, we investigate the structure of the CTF and obtain the following lemma.
\begin{lemma}\label{LXK}
    Given $N$ and $R$, the CTF of a token has $K=\frac{\operatorname{lcm}(N,R)}{N}$ ITTs, in which $\frac{\operatorname{lcm}(N-R\lfloor\frac{N}{R}\rfloor,R)}{R}$ ITTs last for $\lfloor\frac{N}{R}\rfloor+1$ time slots, the remaining ITTs last for $\lfloor\frac{N}{R}\rfloor$ time slots.
    Specifically, the locations of token $\tau$ assigned in the last slot of the consecutive ITTs follow a circular shift process in the ascending direction on the RU domain. 
    We have $\hat{X}_i(k)=\lfloor\frac{N}{R}\rfloor+1$ if and only if the token $\tau$ position shifts through RU $R$.
\end{lemma}
\begin{proof}
    See Appendix \ref{Proof_ITT}.
\end{proof}

Following Lemma \ref{LXK}, we can attain the value of each $\hat{X}_i(k)$.
Thereafter, based on the critical assumption, we can derive a closed-form approximation of $\overline{\Delta}_i$, given in the theorem below.
\begin{theorem}\label{TheoremAAoI}
    Given $N$, $R$, $p_{i,r}$, and $\pi_i=\tau$, the approximation of  $\overline{\Delta}_i$ given by \eqref{AoICTF} and \eqref{QEC} can be expanded into \eqref{AAoIApprox} at the top of the next page, where $\overline{S}_i$ denotes the steady-state service time of sensor $i$, and $\hat{p}_{i,k}$ represents the success rate of the transmission before the $k$th ITT in the CTF of token $\tau$. 
\end{theorem}
\begin{figure*}[t]
    \begin{equation}\label{AAoIApprox}
    \overline{\Delta}_i\!=\!\overline{S}_i\!+\!\left(1\!-\!\frac{R\sum^K_{k=1}\hat{X}_i(k)\prod^k_{j=1}(1-\hat{p}_{i,j})}{\operatorname{lcm}(N,R)}\right)^{-1}\!\left(\frac{R\sum^K_{k=1}\hat{X}^2_i(k)}{2\operatorname{lcm}(N,R)}\!+\!\frac{R}{\operatorname{lcm}(N,R)}\sum^K_{k=1}\hat{X}_i(k)\sum^{k-1}_{z=0}\hat{X}_i(z)\!\prod^k_{j=z+1}(1-\hat{p}_{i,j})\!-\!\frac{1}{2}\right).
\end{equation}
\begin{equation}
    \begin{split}
        \overline{S}_i\!
        =\!\frac{\lambda_i}{(M_i-1)\lambda_i\!+\!1}\!\left(\frac{(1\!+\!U_i)U_i}{2}\!\left\lfloor\frac{M_i}{U_i}\right\rfloor\!+\!\frac{\left(M_i\!-\!U_i\left\lfloor\frac{M_i}{U_i}\right\rfloor\right)\!\left(M_i\!-\!U_i\left\lfloor\frac{M_i}{U_i}\right\rfloor\!+\!1\right)}{2}\!+\!\frac{(1\!-\!\lambda_i)\!\left(\left(M_i\!-\!U_i\!\left\lfloor\frac{M_i-1}{U_i}\right\rfloor\right)\lambda_i\!+\!1\right)}{\lambda_i^2}\right).
    \end{split}
\end{equation}
\vspace{-0.5em}
\hrule
\vspace{-1.5em}
\end{figure*}
\begin{proof}
    See Appendix \ref{Proof_AAoI}.
\end{proof}
We then focus on the derivation of $\overline{S}_i$ shown in \eqref{AAoIApprox}.
Recall that the service time of the status update of a sensor depends on the local state of this sensor.
Hence, we aim to obtain the expression of $\overline{S}_i$ by investigating the local steady state of sensor $i$, which consists of two components, the local age and the active/idle state of sensor $i$.
\begin{figure}[t]
	\centering
	\includegraphics[width=0.48\textwidth]{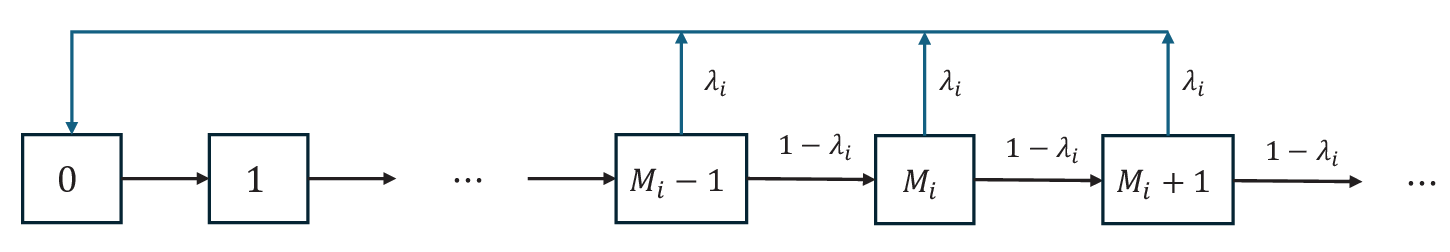}
	\caption{The DTMC of the local state of sensor $i$.}
	\label{LSMC}
\vspace{-1em} 
\end{figure}
We propose a simple-structure DTMC, depicted in Fig. \ref{LSMC}, to explore the evolution of the locate state of sensor $i$. 
Let $A_i$ and $\mu_{A_i}$ denote the state and the associated steady-state probability of the DTMC, respectively.
Specifically, $A_i=0,1,\cdots,M_i-1$ denote the $(A_i+1)$th slot of the active duration of sensor $i$, and states $A_i\ge M_i$ denote $(A_i-M_i+1)$th slots sensor $i$ maintains idle after the active duration.
This DTMC can jointly characterize the evolutions of the two components of the local state of sensor $i$ only by one-dimensional state, thus facilitating the analysis of the local steady state.
According to the structure of the DTMC, we have $\mu_{A_i}=\mu_{A_i+1}$ for $A_i<M_i-1$, while $\mu_{A_i+1}=(1-\lambda_i)\mu_{A_i}$ and $\mu_0=\lambda_i\mu_{A_i}$ for $A_i\ge M_i-1$.
Drawing on this, we can present
\begin{equation}
    \mu_{A_i}=
    \begin{cases}
        \frac{\lambda_i}{(M_i-1)\lambda_i+1},& \text{for }A_i=0,1,\cdots,M_i-1,\\
        \frac{\lambda_i(1-\lambda_i)^{A_i-M_i+1}}{(M_i-1)\lambda_i+1},& \text{otherwise.}
    \end{cases}
\end{equation}
In addition, we denote the associated service time of the status update corresponding to state $A_i$ by $S_i(A_i)$, and its expression can be given by
\begin{equation}
    \begin{split}
       &S_i(A_i)=\\
    &\begin{cases}
        A_i+1-U_i\left\lfloor\frac{A_i}{U_i}\right\rfloor,& \text{for }A_i=0,1,\cdots,M_i-1,\\
        A_i+1-\!U_i\!\left\lfloor\frac{M_i\!-\!1}{U_i}\right\rfloor,& \text{otherwise,}
    \end{cases} 
    \end{split}
\end{equation}
in view of the status update generation pattern.
Consequently, the expression of $\overline{S}_i$ can be achieved following $\overline{S}_i=\sum_{A_i}\mu_{A_i}S_i(A_i)$, shown in at the top of this page.

With Theorem~\ref{TheoremAAoI}, we can assess the AAoI performance of a D-TFDA network given a token allocation policy $\bm{\pi}$.
A natural implication is that different choices of $\bm{\pi}$ lead to different AAoI outcomes, since $\bm{\pi}$ determines which sensor accesses which token and consequently shapes the per-sensor inter-token intervals and the corresponding success rates used in~\eqref{AAoIApprox}.
This observation suggests that the design of the token allocation policy is a lever that can be exploited to improve the network-level information freshness.
We are therefore motivated to develop systematic algorithms for constructing an AAoI-efficient $\bm{\pi}$, which is the focus of the next section.

\section{AoI-Oriented Token Allocation Policy}
\subsection{Problem Formulation and Analysis}
From Theorem~\ref{TheoremAAoI}, the AAoI of each sensor depends on the token allocated to it.
Therefore, the network-wide AAoI performance can be optimized by properly designing the token allocation policy $\bm{\pi}$.
Formally, the AAoI-optimal token allocation problem is formulated as
\begin{subequations}\label{OptProblem}
\begin{align}
    \min_{\bm{\pi}}\quad & \overline{\Delta}^{\bm{\pi}} \triangleq \frac{1}{N}\sum^{N}_{i=1}\overline{\Delta}_i^{\pi_i} \label{OptObj}\\
    \text{s.t.}\quad & \pi_i \neq \pi_{i'},\quad \forall\, i \neq i', \label{OptConst1}\\
    & \pi_i \in \{1,2,\cdots,N\},\quad \forall\, i, \label{OptConst2}
\end{align}
\end{subequations}
where $\overline{\Delta}^{\bm{\pi}}$ denote the AAoI of the D-TFDA network when policy $\bm{\pi}$ is applied, $\overline{\Delta}_i^{\pi_i}$ is the AAoI of sensor $i$ given by~\eqref{AAoIApprox} when we have allocation $\pi_i$, and constraint~\eqref{OptConst1} ensures that each token is assigned to at most one sensor.

Problem~\eqref{OptProblem} is a finite integer optimization problem, for which the global optimum can always be found by exhaustive search.
However, since each sensor must be assigned a distinct token, there are $N!$ possible token allocations in total.
By Stirling's approximation, $N! \sim \sqrt{2\pi N}\left(\frac{N}{e}\right)^N$ \cite{lozier2003nist}, which grows super-exponentially with $N$.
As such, exhaustive search is computationally prohibitive for practical network scales.
We are thus driven to exploit the inherent structure of the network token assignment pattern to reduce the search space and enable efficient optimization.

By examining the network token assignment pattern of D-TFDA, we observe a fundamental structural property that constrains the RUs in which each token can appear, formalized in the following theorem.
\begin{theorem}\label{TheoremMod}
Given $N$ and $R$, a token $\tau \in \{1,2,\cdots,N\}$ appears only in RU $a$, where $a \equiv \tau \pmod{\gcd(N,R)}$.
\end{theorem}
\begin{proof}
See Appendix \ref{ProofTheoremMod}.
\end{proof}
Theorem~\ref{TheoremMod} establishes that the RU locations of each token are fully determined by its index and $N,R$.
Building upon this, we can further characterize the complete set of RUs in which a given token appears, as stated in the following corollary.
\begin{corollary}\label{CorollaryRU}
    A certain token $\tau$ locates in $K$ different RUs, and gaps between any two adjacent RUs of these $K$ RUs are the same, equal to $\gcd(N,R)$.
\end{corollary}
\begin{proof}
    See Appendix \ref{ProofCorollaryRU}.
\end{proof}

Theorem~\ref{TheoremMod} and Corollary~\ref{CorollaryRU} together reveal a structured regularity in the RU occupancy of each token across the network token assignment pattern.
This regularity suggests that not all $N!$ token allocation policies are distinct in terms of their induced AAoI performance, and that a significant portion of the policy space may in fact be redundant.
Along this direction, we establish the following proposition.
\begin{proposition}\label{PropositionAETC}
    In the token assignment array, define all different types of tokens that can appear in the same RUs as a cluster.
    Then, the tokens $1, 2, \cdots, N$ can be partitioned into $\Theta=\gcd(N,R)$ equal-sized token clusters, each comprising $\Gamma=\frac{N}{\Theta}$ tokens and being pairwise disjoint from others.
    Allocating any token within the same cluster to a sensor results in the same AAoI performance for that sensor.
    
    We define these clusters as \textbf{AoI-equivalent token clusters (AETCs)}.
\end{proposition}
\begin{proof}
    See Appendix \ref{ProofPropositionAETC}.
\end{proof}

These structural insights suggest that the vast policy space of $\bm{\pi}$ can be navigated more effectively, potentially overcoming the computational complexity associated with exhaustive search. 
Leveraging these theoretical observations, we now develop systematic and efficient algorithms to determine the AAoI-oriented token allocation policies in the following subsection.

\subsection{Efficient Policy Design}
Let index AETCs as $\theta\in\{1,2,\cdots,\Theta\}$, and denote the indicator that one of token in AETC $\theta$ is allocated to sensor $i$ by $x_{i,\theta}\in\{0,1\}$.
In light of Proposition \ref{PropositionAETC}, any two different $\bm{\pi}$ and $\bm{\pi}^{\prime}$ can be regarded as equivalent solutions if $x_{i,\theta}=1$ hold for some $\theta$ and all $i$ such that $\pi_i\neq\pi^{\prime}_i$.
In particular, we have
\begin{remark}\label{RemarkCoprime}
    When $N$ and $R$ are coprime, we have $\Theta=\gcd(N,R)=1$, hence all types of tokens belong to one AETC. 
    This indicates that the AAoI performance of the D-TFDA network is invariant under any token allocation policy $\bm{\pi}$. In other words, we do not need to design the AoI-oriented token allocation in this context.
\end{remark}
To be generalized, it follows that the search space of Problem~\eqref{OptProblem} can be reduced by considering only one representative solution among all equivalent candidates.
After the space reduction, Problem \eqref{OptProblem} is reformulated as
\begin{subequations}\label{OptProblemReduced}
\begin{align}
    \min_{\bm{\pi}}\quad 
    & \overline{\Delta}^{\bm{\pi}}
    = \frac{1}{N}\sum^{N}_{i=1}\overline{\Delta}_i^{\pi_i} \\
    \text{s.t.}\quad 
    & \pi_i \neq \pi_{i'},\quad \forall\, i \neq i', \\
    & \pi_i \in \{1,2,\cdots,N\},\quad \forall\, i, \\
    & \pi_i < \pi_{i'},\quad \forall\, i < i'\text{ such that}\nonumber\\
    &x_{i,\theta}=x_{i',\theta}=1,\text{ for some }\theta,\label{order}
\end{align}
\end{subequations}
where constraint \eqref{order} imposes a symmetry-breaking ordering rule to eliminate redundant equivalent solutions.
Specifically, since each group contains $\Gamma!$ equivalent permutations and there exist $\Theta$ AETCs, the cardinality of the reduced solution space becomes $\frac{N!}{(\Gamma!)^\Theta}$, which is significantly smaller than the original solution space size $N!$.

By delving into the structure of Problem \eqref{OptProblemReduced}, we find this problem is equivalent to the following min-cost flow problem, presented as
\begin{subequations}\label{MinCostFlowProblem}
\begin{align}
    \min_{\{x_{i,\theta}\}} \quad 
    & \sum_{\theta=1}^{\Theta}\sum_{i=1}^{N} x_{i,\theta}\,\overline{\Delta}_i^{\theta} \\
    \text{s.t.}\quad
    & \sum_{i=1}^{N} x_{i,\theta} = \Gamma,
    \quad \forall\, \theta, \\
    & \sum_{\theta=1}^{\Theta} x_{i,\theta} = 1,
    \quad \forall\, i, \\
    & x_{i,\theta}\in\{0,1\},
    \quad \forall\, i,\theta,
\end{align}
\end{subequations}
where $\overline{\Delta}_i^{\theta}$ denotes AAoI of sensor $i$ when $x_{i,\theta}=1$.
While Problem~\eqref{MinCostFlowProblem} is an integer program, its specific structure allows for a more computationally efficient approach. 
Specifically, by relaxing the binary constraint $x_{i,\theta}\in\{0,1\}$ into $0 \le x_{i,\theta} \le 1$, Problem~\eqref{MinCostFlowProblem} can be converted to the following standard linear programming (LP) problem:
\begin{subequations}\label{LPProblem}
\begin{align}
    \min_{\bm{x}}\quad 
    & \bm{c}^{\top}\bm{x} \\
    \text{s.t.}\quad
    & \mathbf{A}\bm{x} = \bm{b}, \\
    & \bm{x}\ge \bm{0},
\end{align}
\end{subequations}
where
\begin{equation}
    \begin{split}
    \bm{c}=
    &\big[
    \overline{\Delta}^{1}_{1},
    \overline{\Delta}^{2}_{1},
    \cdots,
    \overline{\Delta}^{\Theta}_{1},
    \overline{\Delta}^{1}_{2},
    \overline{\Delta}^{2}_{2},
    \cdots,
    \overline{\Delta}^{\Theta}_{2},
    \cdots,\\
    &\overline{\Delta}^{1}_{N},
    \overline{\Delta}^{2}_{N},
    \cdots,
    \overline{\Delta}^{\Theta}_{N}
    \big]^{\top},
\end{split}
\end{equation}
\begin{equation}
    \begin{split}
    \bm{x}=&\big[x_{1,1},x_{1,2},\cdots,x_{1,\Theta},x_{2,1},x_{2,2},\cdots,x_{2,\Theta},\cdots,\\
    &x_{N,1},x_{N,2},\cdots,x_{N,\Theta}\big]^{\top},
\end{split}
\end{equation}
and
\begin{align}
    \mathbf{A}
=
\begin{array}{@{}c@{}l@{}}
\left[
\begin{matrix}
    \bm{1}_{\Theta} &        &        &   \\
                    & \bm{1}_{\Theta} &  &   \\
                    &        & \ddots &   \\
                    &        &        & \bm{1}_{\Theta} \\
    \mathbf{I}_{\Theta} & \mathbf{I}_{\Theta} & \cdots & \mathbf{I}_{\Theta}
\end{matrix}
\right]
&
\mkern-6mu
\begin{array}{@{}l@{}}
\left.
\begin{array}{@{}c@{}}
\vphantom{\bm{1}_{\Theta}}\\
\vphantom{\bm{1}_{\Theta}}\\
\vphantom{\ddots}\\
\vphantom{\bm{1}_{\Theta}}
\end{array}
\right\}\,N\\
\vphantom{\mathbf{I}_{\Theta}}
\end{array}
\\[-0.8em]
\underbrace{
\smash{
\hphantom{
\begin{matrix}
    \bm{1}_{\Theta} &        &        &   \\
                    & \bm{1}_{\Theta} &  &   \\
                    &        & \ddots &   \\
                    &        &        & \bm{1}_{\Theta} \\
    \mathbf{I}_{\Theta} & \mathbf{I}_{\Theta} & \cdots & \mathbf{I}_{\Theta}
\end{matrix}
}
}
}_{N\Theta}
&
\end{array},
    \quad
    \bm{b}
=
\begin{bmatrix}
    1 \\
    1 \\
    \vdots \\
    1 \\
    \Gamma \\
    \Gamma \\
    \vdots \\
    \Gamma
\end{bmatrix}\mkern-8mu
\begin{array}{@{}l@{}}
\left.
\begin{array}{@{}c@{}}
\vphantom{1}\\
\vphantom{1}\\
\vphantom{\vdots}\\
\vphantom{1}
\end{array}
\right\}\,N\\
\left.
\begin{array}{@{}c@{}}
\vphantom{\Gamma}\\
\vphantom{\Gamma}\\
\vphantom{\vdots}\\
\vphantom{\Gamma}
\end{array}
\right\}\,\Theta
\end{array},
\end{align}
with $\bm{1}_{\Theta}$ denoting the all-one row vector of length $\Theta$, and $\mathbf{I}_{\Theta}$ denoting the $\Theta$-order identity matrix.

Although Problem~\eqref{LPProblem} applying the relaxation is slightly different from the original problem, the following remark establishes that this relaxation incurs no gap in optimality.
\begin{remark}\label{RemarkTUM}
The relaxation is without loss of optimality.
Since the constraint matrix $\mathbf{A}$ of Problem~\eqref{LPProblem} is totally unimodular and the right-hand side vector $\bm{b}$ is integer-valued, the Hoffman--Kruskal theorem \cite{hoffman2009integral} guarantees that every basic feasible solution of the LP is integral.
Consequently, solving the linear program~\eqref{LPProblem} yields an optimal solution that is automatically binary, thereby exactly solving the original integer program~\eqref{MinCostFlowProblem}.
\end{remark}
Built upon this groundwork, we develop an efficient algorithm that solves Problem~\eqref{OptProblem} to optimality, with the circumvention of the combinatorial complexity of the original search space, yielding the AAoI-optimal token allocation policy $\bm{\pi}^*$. The algorithm is outlined in Algorithm \ref{LPAlgOptAlloc}.
\begin{algorithm}[t]
\SetAlgoLined
\KwIn{$N$, $R$, $\{M_i\}_{i=1}^{N}$, $\{U_i\}_{i=1}^{N}$, $\{p_{i,r}\}_{i=1,r=1}^{N,R}$}
Compute $\Theta$, $\Gamma$, and partition the $N$ tokens into $\Theta$ AETCs via Proposition~\ref{PropositionAETC}\;
\For{$\theta = 1$ \KwTo $\Theta$}{
    Select a representative token $\tau_{\theta}$ from AETC $\theta$\;
    \For{$i = 1$ \KwTo $N$}{
        Compute $\overline{\Delta}_i^{\theta}$ via~\eqref{AAoIApprox} with token $\tau_\theta$\;
    }
}
Construct cost vector $\bm{c} \in \mathbb{R}^{N\Theta}$ by stacking all $\overline{\Delta}_i^{\theta}$\;
Solve LP~\eqref{LPProblem} to obtain $\bm{x}^*$\;
Recover $\hat{\bm{\pi}}^*$ from $\bm{x}^*$ as the nearest feasible binary solution\;
\Return{$\hat{\bm{\pi}}^*$}\;
\caption{LP-based AAoI-Optimal Policy Search}\label{LPAlgOptAlloc}
\end{algorithm}

While Algorithm~\ref{LPAlgOptAlloc} already reduces the search space dramatically relative to the exhaustive search method, it relies on solving the LP Problem~\eqref{LPProblem}, which has variable $\bm{x}$ of dimension $N\Theta$ and $N+\Theta$ constraints. 
When $N\Theta$ and $N+\Theta$ grow large, the resulting LP instance may incur considerable computational complexity and encounter memory bottlenecks that hinder their applicability to resource-constrained or large-scale deployments. 
These practical limitations motivate us to pursue an alternative algorithmic strategy that trades off a marginal degree of optimality for substantially improved efficiency. 
In the following, we propose a lightweight heuristic algorithm capable of producing an AoI-effective token allocation with significantly lower computational and memory overhead.

It is worth noting that Problem~\eqref{MinCostFlowProblem} bears a close structural resemblance to the classical assignment problem. 
In particular, the $\Theta$ AETCs can be viewed as $\Theta$ types of goods each with $\Gamma$ identical copies, while the $N$ sensors act as bidders that each demand exactly one good with a cost $\overline{\Delta}_i^{\theta}$. 
This interpretation naturally invites an auction-theoretic solution paradigm~\cite{bertsekas1988auction}.
Inspired by this connection, we propose an efficient heuristic token allocation algorithm that applies the auction mechanism, shown in Algorithm \ref{AlgHeuristic}.
\begin{algorithm}[t]
\SetAlgoLined
\KwIn{$\varepsilon > 0$, $N$, $R$, $\{M_i\}_{i=1}^{N}$, $\{U_i\}_{i=1}^{N}$, $\{p_{i,r}\}_{i=1,r=1}^{N,R}$}
Compute $\Theta$, $\Gamma$, partition the $N$ tokens into $\Theta$ AETCs, and calculate $\{\overline{\Delta}^\theta_i\}^{N,\Theta}_{i=1,\theta=1}$ as that in Algorithm \ref{LPAlgOptAlloc}\;
Initialize extra cost vector $\bm{E} \leftarrow \bm{0} \in \mathbb{R}^{\Theta}$, the set of unallocated sensors $\mathcal{N} \leftarrow \{1, 2, \dots, N\}$, and allocation lists $\mathcal{L}_\theta \leftarrow \emptyset$ for all $\theta \in \{1, \dots, \Theta\}$\;
\While{$\mathcal{N} \neq \emptyset$}{
    Select an arbitrary sensor $i \in \mathcal{N}$ and set $\mathcal{N} \leftarrow \mathcal{N} \setminus \{i\}$\;
    $\theta^* \leftarrow \arg\min_{\theta} \big( \overline{\Delta}_i^{\theta} + E_{\theta} \big)$\;
    $V_1 \leftarrow \overline{\Delta}_i^{\theta^*} + E_{\theta^*}$\;
    $V_2 \leftarrow \min_{\theta \neq \theta^*} \big( \overline{\Delta}_i^{\theta} + E_{\theta} \big)$\;
    $E_{\theta^*} \leftarrow E_{\theta^*} + V_2 - V_1 + \varepsilon$\;
    $\mathcal{L}_{\theta^*} \leftarrow \mathcal{L}_{\theta^*} \cup \{i\}$\;
    \If{$|\mathcal{L}_{\theta^*}| > \Gamma$}{
        Let $i'$ be the earliest admitted sensor in $\mathcal{L}_{\theta^*}$\;
        $\mathcal{L}_{\theta^*} \leftarrow \mathcal{L}_{\theta^*} \setminus \{i'\}$\;
        $\mathcal{N} \leftarrow \mathcal{N} \cup \{i'\}$\;
    }
}
Obtain $\tilde{\bm{\pi}}^*$ from $\{\mathcal{L}_\theta\}_{\theta=1}^{\Theta}$\;
\Return{$\tilde{\bm{\pi}}^*$}\;
\caption{Heuristic Policy Search}\label{AlgHeuristic}
\end{algorithm}

In terms of computational complexity, Algorithm~\ref{AlgHeuristic} is significantly more efficient than the LP-based policy search in Algorithm~\ref{LPAlgOptAlloc}. 
Specifically, solving the LP problem by the simplex method generally incurs a complexity of $\mathcal{O}(N\Theta(N+\Theta)^2)$, while interior-point methods typically require complexity ranging from $\mathcal{O}((N+\Theta)^{1.5})$ to $\mathcal{O}((N+\Theta)^3)$ \cite{REHFELDT202260}. 
In contrast, the auction-based heuristic only involves iterative local bidding and assignment updates, leading to a much lower complexity of approximately $\mathcal{O}(N+\Theta)$ \cite{bertsekas1992auction}. 
Therefore, Algorithm~\ref{AlgHeuristic} achieves substantially lower computational overhead and is more suitable for large-scale implementations.

The performance of the proposed algorithms is examined through extensive numerical simulations in the next section.

\section{Numerical Results}
In this section, we first verify the validity of our theoretical analysis through a comparison between theoretical derivations and simulation results. Subsequently, we evaluate the effectiveness of the proposed token allocation policy search algorithms. Furthermore, we assess the overall performance of the D-TFDA network by comparing its AAoI performance against baseline transmission schemes. For a fair and statistically significant evaluation, all performance curves presented in the following figures are obtained by averaging over $10^6$ Monte-Carlo simulation runs.
In addition, $p_{i,r}$ is generated using a bounded random reliability model. 
Specifically, each STA has a mean reliability $\bar p_i\sim{\rm Beta}(\alpha,\beta)$, while correlated Gaussian perturbations with $\rho$ and $\sigma$ are added across neighboring RUs and then clipped into $[0,1]$. 
This model captures STA-level link heterogeneity and RU-level frequency selectivity, while using the Beta distribution to model bounded packet-level link reliability, as commonly adopted for wireless packet success probability and link adaptation approximation~\cite{9154253,9462466}.

\subsection{Verifications of the Analysis}
\begin{figure}[t]
	\centering
	\includegraphics[width=0.44\textwidth]{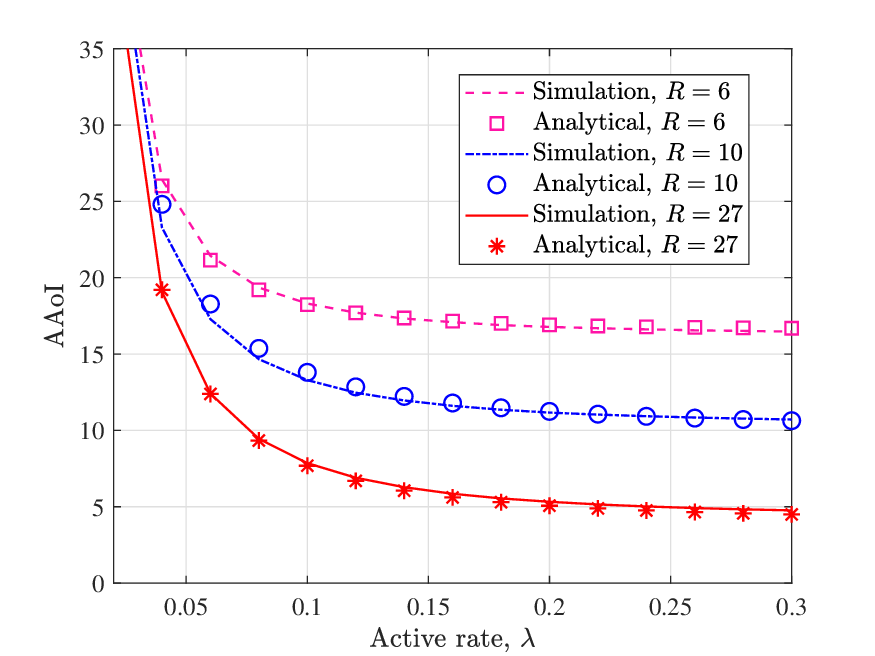}
	\caption{AAoI versus the active rate, with $N=100$ and $\lambda_i=\lambda,\forall i$.}
 \label{AAoI_v_la}
 \vspace{-1em}
\end{figure}
Fig.~\ref{AAoI_v_la} illustrates the AAoI performance versus the active rate $\lambda$, where $N=100$, all sensors have the same active rate, i.e., $\lambda_i=\lambda$, $\forall i$, $\alpha=8$, $\beta=2$, $\rho=0.6$, $\sigma=0.12$, and $\bm{\pi}=[1,2,\cdots,N]$. 
Three system settings are considered: $R=6$, $M_i\sim \mathcal U_{\mathbb Z}[32,36]$, and $U_i\sim \mathcal U_{\mathbb Z}[6,8]$; $R=10$, $M_i\sim \mathcal U_{\mathbb Z}[20,24]$, and $U_i\sim \mathcal U_{\mathbb Z}[3,5]$; and $R=27$, $M_i\sim \mathcal U_{\mathbb Z}[14,18]$, and $U_i\sim \mathcal U_{\mathbb Z}[2,4]$. 
It is shown that the analytical results closely match the simulation results under all considered settings, with the relative error below $5\%$, further verifying the accuracy of the theoretical AAoI approximation. 
In addition, the AAoI decreases as $\lambda$ increases, since a larger active rate allows sensors to generate update packets more frequently and thus improves information freshness. 
The AAoI reduction is significant when $\lambda$ is small, while the improvement gradually becomes marginal as $\lambda$ further increases.
This is because, in the low-active-rate regime, the intervals between two consecutive active durations are relatively long, so increasing $\lambda$ effectively shortens the inactive periods during which no fresh updates can be generated. 
In contrast, when $\lambda$ becomes sufficiently large, sensors are already active frequently, and the AAoI is mainly constrained by the update generation interval, the available transmission opportunities, and the success rate.

\begin{figure}[t]
	\centering
	\includegraphics[width=0.44\textwidth]{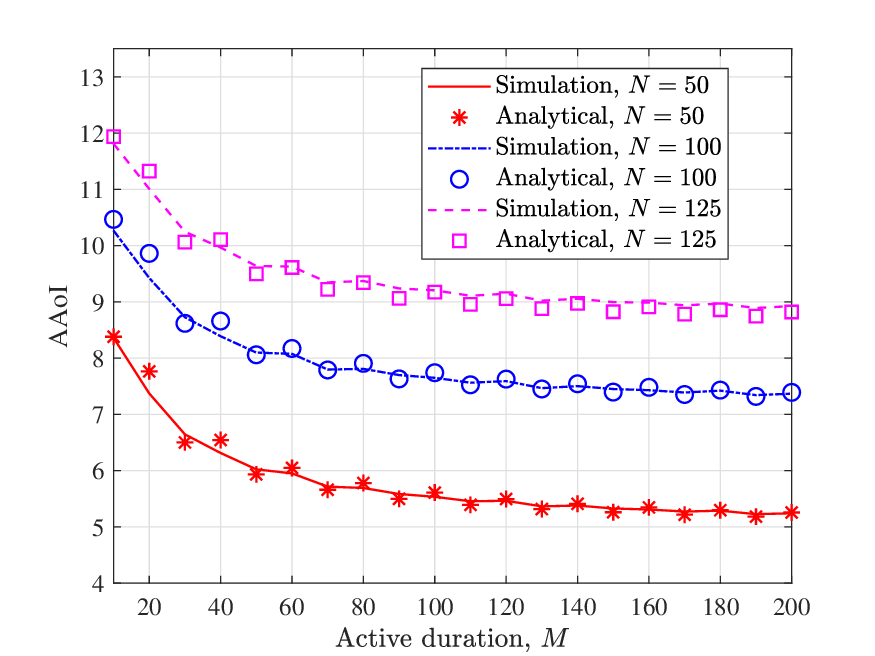}
	\caption{AAoI versus the active duration, with $R=15$ and $\lambda_i=0.12,\forall i$.}
 \label{AAoI_v_M}
 \vspace{-1em}
\end{figure}
Fig.~\ref{AAoI_v_M} depicts the AAoI as a function of the active duration, where part of the simulation setup is specified in the caption.
In all cases, we set $U_i\sim\mathcal U_{\mathbb Z}[3,5]$, $M_i=M$, and $\bm{\pi}=[1,2,\cdots,N]$, the parameters of the Beta distribution for generations of $p_{i,r}$ are the same as that in Fig.~\ref{AAoI_v_la}, and consider three network sizes with $N=50$, $N=100$, and $N=125$.
Similar to Fig.~\ref{AAoI_v_la}, the gap between the analytical and simulation results remains within $5\%$ under all considered settings.
It is observed that the AAoI decreases as $M$ increases, since a longer active duration allows each active sensor to generate more update packets within one active period and thus provides more opportunities for successful status delivery.
The slight oscillations in the AAoI curve are intuitive, since when $M$ is not an integer multiple of $U_i$, extending the active duration may introduce additional slots without new status update generation, which provides limited freshness improvement.
Once $M$ reaches the next multiple of $U_i$, one more update can be generated and potentially transmitted, thereby tending to reduce the AAoI and producing the observed fluctuation.
Moreover, a larger value of $N$ leads to a higher AAoI, since more sensors share the same RU resources and each sensor obtains transmission opportunities less frequently under the same token allocation structure.

\begin{figure}[t]
	\centering
	\includegraphics[width=0.44\textwidth]{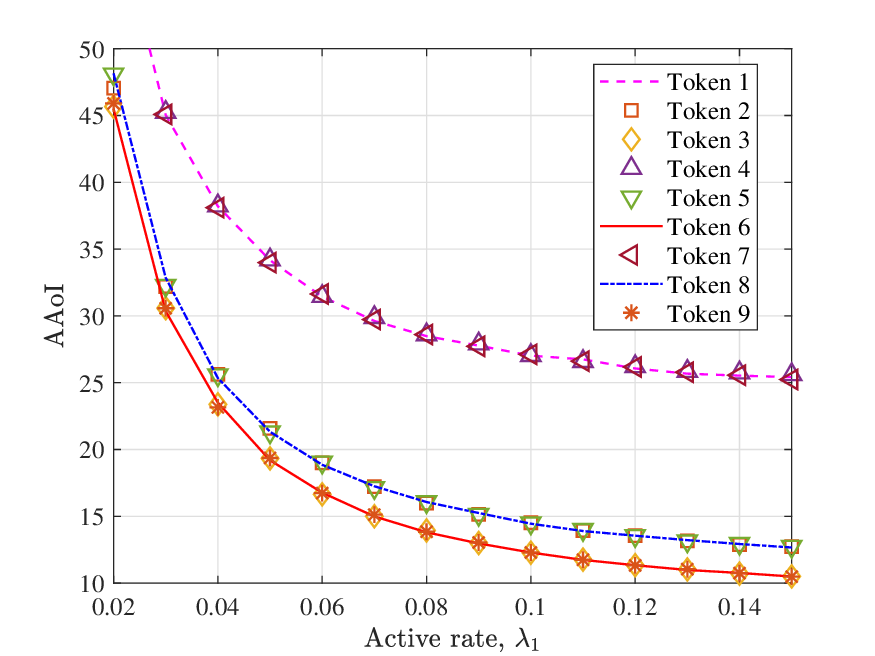}
	\caption{AAoIs of sensor $1$ allocated with all tokens, with $N=9$ and $R=6$.}
 \label{AAoI_Diff_TAs}
 \vspace{-1em}
\end{figure}
Fig.~\ref{AAoI_Diff_TAs} shows the AAoI performance of sensor $1$ when it is assigned different tokens, where part of the simulation setup is given in the caption.
In Fig.~\ref{AAoI_Diff_TAs}, the packet success probabilities are generated with $\alpha=2$, $\beta=2$, $\rho=0.2$, and $\sigma=0.24$, while $M_i\sim\mathcal U_{\mathbb Z}[14,18]$ and $U_i\sim\mathcal U_{\mathbb Z}[2,4]$.
According to Proposition~\ref{PropositionAETC}, the network has $\Theta=3$ AETCs, corresponding to the token sets $\{1,4,7\}$, $\{2,5,8\}$, and $\{3,6,9\}$.
It can be seen that, within each AETC, the AAoI curves associated with different tokens are almost identical over the considered range of $\lambda_1$.
This confirms that tokens belonging to the same AETC yield equivalent AAoI performance for the same sensor, thereby validating Proposition~\ref{PropositionAETC}.

\begin{figure}[t]
	\centering
	\includegraphics[width=0.44\textwidth]{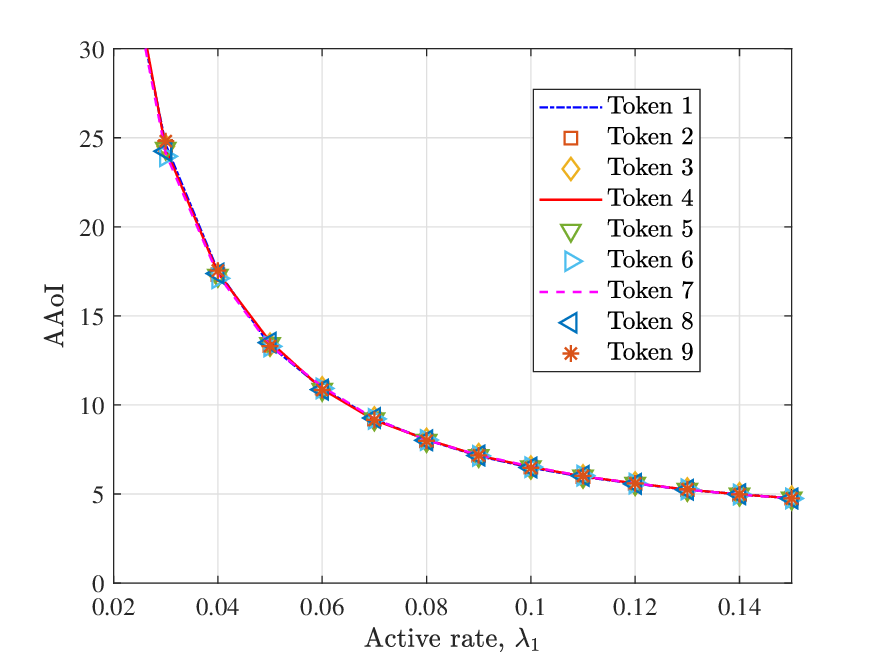}
	\caption{AAoIs of sensor $1$ allocated with all tokens, with $N=9$ and $R=5$.}
 \label{AAoI_Diff_TAs_2}
 \vspace{-1em}
\end{figure}
Fig.~\ref{AAoI_Diff_TAs_2} illustrates AAoI of sensor $1$ under different token allocations, using the same simulation setup as Fig.~\ref{AAoI_Diff_TAs} except that $R$ is changed to $5$.
Since $N=9$ and $R=5$ are coprime, all tokens belong to a single AETC.
As shown in the figure, the AAoI curves corresponding to all nine tokens almost completely overlap over the whole range of $\lambda_1$, following the statement given in Remark~\ref{RemarkCoprime}.

\subsection{Performance Evaluation and Comparison}
We first verify the effectiveness of the proposed token allocation policy search algorithms through the following numerical comparison.

\begin{figure}[t]
	\centering
	\includegraphics[width=0.44\textwidth]{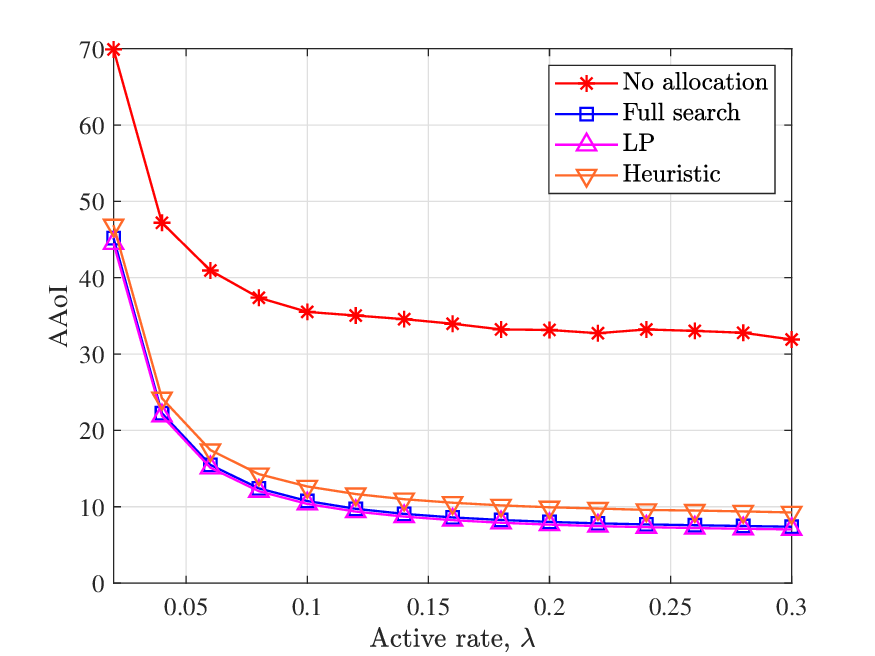}
	\caption{AAoIs of the D-TFDA network applying different token allocation policies, with $N=15$ and $R=6$.}
 \label{AAoI_Diff_TAPs}
 \vspace{-1em}
\end{figure}
Fig.~\ref{AAoI_Diff_TAPs} compares the AAoI performance of the D-TFDA network adopting different token allocation policy search algorithms as the increase of the active rate.
We set $M_i\sim\mathcal U_{\mathbb Z}[5,25]$, $U_i\sim\mathcal U_{\mathbb Z}[1,7]$, and $\lambda_i=\lambda$, $\forall i$.
$p_{i,r}$ are generated with $\alpha=2$, $\beta=2$, $\rho=0.2$, and $\sigma=0.24$, and $\varepsilon=10^{-3}$.
The ``no allocation'' baseline directly allocates token $i$ to sensor $i$ without optimizing the token allocation policy.
The full search algorithm exhaustively examines all $N!/(\Gamma!)^\Theta$ candidate policies and is therefore used as the optimal benchmark.
The comparison shows that both proposed token allocation policy search algorithms substantially outperform the no-allocation baseline over the entire range of $\lambda$.
In particular, the AAoI achieved by the LP-based policy search almost coincides with that of the full search benchmark, indicating that it can identify the optimal token allocation policy with reduced search complexity.
The heuristic policy search yields a slightly higher AAoI than the full search benchmark, but its performance remains close to the optimum and clearly superior to the no-allocation baseline.
These results reveal the effectiveness of the proposed token allocation policy search algorithms in improving the AAoI performance of the D-TFDA network.

\begin{table}[t]
\centering
\caption{Proposed Algorithm Deviation From Optimum.}
\label{tab:Search_Gap}
\vspace{-0.5em}
\scriptsize
\setlength{\tabcolsep}{3pt}
\renewcommand{\arraystretch}{1.15}
\begin{tabular}{|c|c|c|}
\hline
\textbf{Network setup}
&
\begin{tabular}{c}
\textbf{Full search vs.} \\
\textbf{LP-based search}
\end{tabular}
&
\begin{tabular}{c}
\textbf{Full search vs.} \\
\textbf{Heuristic search}
\end{tabular}
\\
\hline
$N=12$, $R=3$ & $0.2\%$   & $29.07\%$ \\
\hline
$N=14$, $R=4$ & $11.39\%$ & $25.39\%$ \\
\hline
$N=15$, $R=6$ & $15.25\%$ & $34.69\%$ \\
\hline
\end{tabular}
\vspace{-1em}
\end{table}
Table~\ref{tab:Search_Gap} compares the token allocation policies obtained by the proposed search algorithms with the optimal policy obtained by full search.
The random generation models of $M_i$, $U_i$, and $p_{i,r}$ are the same as those used in Fig.~\ref{AAoI_Diff_TAPs}, and we set $\lambda_i=\lambda=0.08$, $\forall i$.
For a token allocation policy $\bm{\pi}$, let $\mathcal S_\theta(\bm{\pi})$ denote the set of sensors assigned to the tokens in the $\theta$-th ETC.
Given the optimal policy $\bm{\pi}^\star$, the policy deviation is defined as $\frac{1}{N}\sum_{\theta=1}^{\Theta}|\mathcal S_\theta(\bm{\pi})\setminus \mathcal S_\theta(\bm{\pi}^\star)|$, which measures the normalized number of sensors assigned to different ETCs compared with the optimal policy.
Table~\ref{tab:Search_Gap} presents the average policy deviation over $1000$ independent trials.
The results show that the policies obtained by the proposed LP-based policy search are consistently closer to the optimal policies than those obtained by the proposed heuristic policy search.
This indicates that the LP-based method can more accurately recover the optimal token allocation structure.

\begin{figure}[t]
	\centering
	\includegraphics[width=0.44\textwidth]{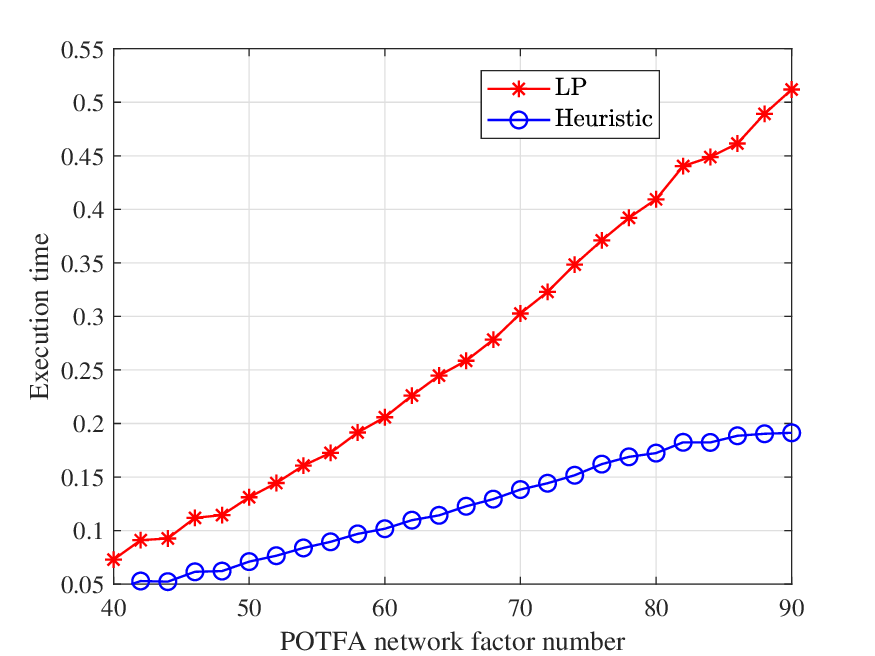}
	\caption{Execution times of the proposed policy search algorithms.}
 \label{Execution_Time}
 \vspace{-1em}
\end{figure}
Fig.~\ref{Execution_Time} plots the average execution times of the proposed LP-based policy search and heuristic policy search algorithms.
The simulations are conducted using MATLAB on a computer equipped with a $13$th Gen Intel(R) Core(TM) i7-13700K processor with $16$ cores and $32$GB of RAM.
The random generation models of $M_i$, $U_i$, and $p_{i,r}$ are the same as those used in Fig.~\ref{AAoI_Diff_TAPs}.
To evaluate the computational efficiency under different network scales, we introduce a network factor number $\Omega$ and set $N=5\Omega$ and $R=2\Omega$.
For each value of $\Omega$, the execution time is measured by recording the CPU running time, in seconds, required by MATLAB to generate the token allocation policy.
Each point in the figure is obtained by averaging the execution times over $1000$ independent runs.
We find that the execution times of both algorithms increase with $\Omega$, since a larger network introduces a larger token allocation search space.
Moreover, the heuristic policy search consistently requires less execution time than the LP-based policy search, and the computational advantage becomes more evident as $\Omega$ increases.
This confirms the remark on the complexity of the two proposed approaches that the heuristic method provides a more computationally efficient alternative for large-scale D-TFDA networks.

\begin{figure}[t]
	\centering
	\includegraphics[width=0.44\textwidth]{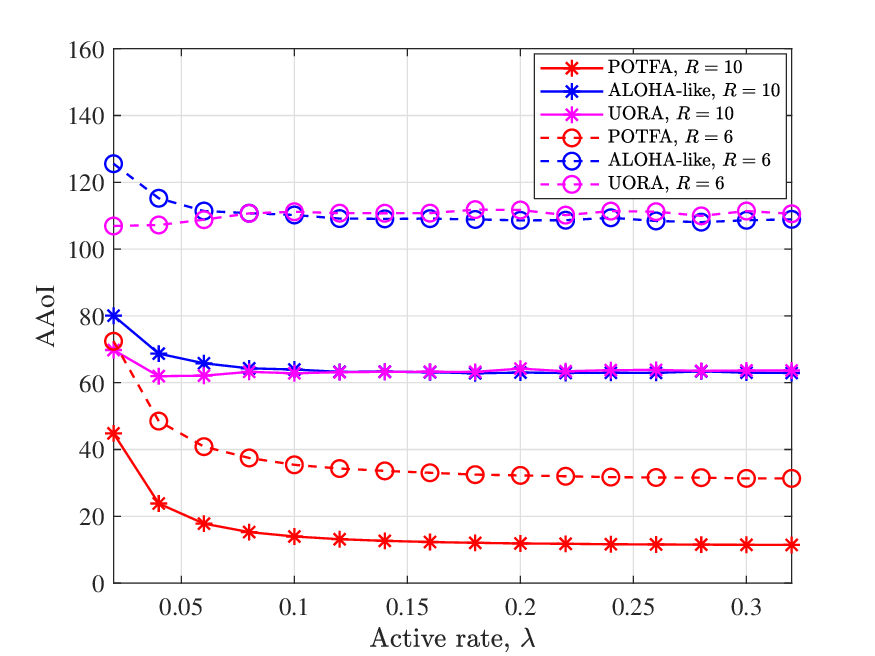}
	\caption{Comparisons of the AAoI between D-TFDA and baseline mechanisms, with $N=100$, $\lambda_i=\lambda,\forall i$.}
 \label{Baseline_Compare}
 \vspace{-1em}
\end{figure}
We further compare the proposed D-TFDA mechanism with baseline random access mechanisms to evaluate its AAoI advantage over contention-based uplink access.
\begin{itemize}
    \item \textbf{ALOHA-like mechanism:}
    Since existing AoI-oriented ALOHA mechanisms are not directly tailored to the considered multi-RU scenario, we adopt a simplified ALOHA-like baseline inspired by slotted ALOHA \cite{1146521}.
    Specifically, whenever a sensor generates a new status update, it randomly selects one RU with probability $\eta$ and transmits the update to the AP.
    The access probability $\eta$ is optimized by the genetic algorithm to obtain the minimum AAoI of this baseline.

    \item \textbf{UORA mechanism:}
    We also consider the uplink OFDMA-based random access (UORA) mechanism in WiFi 6/7, which adopts an OFDMA backoff (OBO) procedure within the OFDMA contention window (OCW) range to allow STAs to contend for random-access RUs announced by the AP.
    The UORA mechanism is characterized by the maximum retransmission number $m$ and the minimum OFDMA contention window ${\rm OCW}_{\min}$.
    Readers can be referred to the 802.11ax amendment \cite{9442429} for a more comprehensive understanding.
    In the simulations, we search over different pairs of $m$ and ${\rm OCW}_{\min}$ and use the best achieved AAoI as the optimized UORA performance.
\end{itemize}
In Fig.~\ref{Baseline_Compare}, we compare the AAoI performance of D-TFDA with the ALOHA-like and UORA mechanisms under two network settings.
In both settings, we have $N=100$, $\lambda_i=\lambda$, $\forall i$, $U_i\sim\mathcal U_{\mathbb Z}[2,4]$, and the packet success probabilities $p_{i,r}$ are generated using the same model as in Fig.~\ref{AAoI_Diff_TAPs}.
For the first setting, we set $R=10$ and $M_i\sim\mathcal U_{\mathbb Z}[23,24]$, while for the second setting, we set $R=6$ and $M_i\sim\mathcal U_{\mathbb Z}[8,10]$.
The token allocation policy of D-TFDA is obtained by the heuristic policy search algorithm.
The comparison reflects that D-TFDA achieves a much lower AAoI than both baseline mechanisms over the entire range of $\lambda$.
This improvement mainly comes from the token-based scheduled transmission structure of D-TFDA, which avoids collisions and provides more regular uplink opportunities for status update delivery.
Moreover, D-TFDA explicitly exploits the heterogeneous transmission success rates of different sensors over the RUs when optimizing the token allocation, whereas the baseline mechanisms optimize only network-wide common parameters.

\section{Conclusion}
In this paper, we developed an age of information (AoI)-oriented deterministic time-frequency distributed access (D-TFDA) mechanism that combines centralized configuration with distributed operation through a periodic and collision-free token-based access structure.
By exploiting the periodicity of the token assignment and modeling each sensor's steady local state using a one-dimensional discrete-time Markov chain (DTMC), we derived a tractable analytical approximation for evaluating the long-term average AoI (AAoI).
We further identified AoI-equivalent token clusters, which substantially reduce the token allocation search space, by exploring the transmission pattern of the D-TFDA.
Based on this structure, we developed a linear programming-based optimal policy search algorithm and an auction-inspired heuristic algorithm with lower computational and memory requirements.
Simulation results validated the proposed analysis and demonstrated that the two algorithms effectively improve the AAoI, while D-TFDA substantially outperforms random access baselines.

Future extensions of this work may investigate AoI analysis and optimization under more general periodic transmission patterns beyond the specific periodic token assignment structure adopted by D-TFDA.
Another interesting direction is to develop AoI-oriented analytical and optimization frameworks for sequence-based distributed access, where sensors determine their transmission opportunities according to individually assigned protocol sequences.

\bibliographystyle{IEEEtran}
\bibliography{Ref}

\appendices

\section{Proof of Lemma \ref{Period}}\label{Proof_Period}
At the beginning of the network operation, i.e., slot $1$, the token assignments in this slot is $1,2,\cdots,R$.
After one period, this token combination is assigned again for the first time.
According to the pattern of the token assignment, the largest TB in the previous time slot is assigned to token $N$, indicating that there are an integer multiple of $N$ TBs in one period.
Moreover, it is clear that there are an integer multiple of $R$ TBs in one period.
Thus, one period includes $\operatorname{lcm}(N,R)$ TBs.
Recalling that each slot has $R$ TBs, we can conclude that the period is $\frac{\operatorname{lcm}(N,R)}{R}$.

\section{Proof of Lemma \ref{LXK}}\label{Proof_ITT}
Consider a CTF corresponding to token $\tau$, which is allocated to sensor $i$, starting from slot $t>1$.
Recall that each ITT in the CTF follows a time slot containing token $\tau$.
Hence, by Lemma \ref{Period}, $K$ equals the number of TBs assigned with token $\tau$ within the interval from slot $t-1$ to slot $t+\frac{\operatorname{lcm}(N,R)}{R}-2$.
Since this interval includes $\operatorname{lcm}(N,R)$ TBs and one of every $N$ TBs, in their order, have one assignment of token $\tau$, we have $K=\frac{\operatorname{lcm}(N,R)}{N}$.
Moreover, this interval consists of $K$ complete sets of $N$ consecutive TBs, that are sequentially assigned with tokens 
\begin{equation}\label{SSNTA}
    (\lceil\frac{\tau}{R}\rceil-1)R +1,(\lceil\frac{\tau}{R}\rceil-1)R +2,\cdots,N,1,\cdots,(\lceil\frac{\tau}{R}\rceil-1)R.
\end{equation}
We define these token assignments as $\tau$-starting-slot network token assignments ($\tau$-SSNTA).
Notice that the values of $\hat{X}_i(k)$'s depend on how in the dimension of the TBs the corresponding $\tau$-SSNTA locates, defined as the $\tau$-SSNTA pattern.

When $N$ is divisible by $R$, it is evident that we only have one $\tau$-SSNTA pattern, i.e., tokens
presented in \eqref{SSNTA},
are assigned to all TBs from slot $t-1$ to slot $t+\frac{N}{R}-2$, $\frac{N}{R}$ time slots in total.
In this context, the CTF has $K=1=\frac{\operatorname{lcm}(N,R)}{N}$ ITT with the duration of $\hat{X}(1)=\frac{N}{R}=\lfloor\frac{N}{R}\rfloor$ time slots, and has $0=\frac{\operatorname{lcm}(N-R\lfloor\frac{N}{R}\rfloor,R)}{R}$ ITT with the duration of $\lfloor\frac{N}{R}\rfloor+1$.
This satisfies Theorem \ref{LXK}.

When $N$ is not divisible by $R$, we have $K$ different $\tau$-SSNTA patterns from slot $t-1$ to slot $t+\frac{\operatorname{lcm}(N,R)}{R}-2$.
To proceed, we assume that in the first $\tau$-SSNTA pattern, token $\tau$ is assigned to $\langle t;1\rangle$.
Then, token $\tau$ in the next $\tau$-SSNTA pattern is assigned to $\langle t+\lfloor\frac{N}{R}\rfloor;N-R\lfloor\frac{N}{R}\rfloor+1 \rangle$.
As such, the duration of the ITT between these two tokens $\tau$ is $\lfloor\frac{N}{R}\rfloor$ time slots.
We refer to the ITTs with this duration as \textbf{type 1} ITTs. 
Subsequently, assuming that in another $\tau$-SSNTA pattern, token $\tau$ is assigned to $\langle t^{'};r \rangle$, the position of token $\tau$ of this $\tau$-SSNTA pattern shifts $r-1$ RUs compared to that of token $\tau$ of the previous $\tau$-SSNTA pattern in the RU domain.
This implies that token $\tau$ of the next $\tau$-SSNTA pattern of this one should be assigned to $r-1$ TBs after $\langle t^{'}+\lfloor\frac{N}{R}\rfloor;N-R\lfloor\frac{N}{R}\rfloor+1 \rangle$.
Clearly, if $N-R\lfloor\frac{N}{R}\rfloor+1+r-1=N-R\lfloor\frac{N}{R}\rfloor+r\le R$, the next token $\tau$ is assigned to $\langle   t^{'}+\lfloor\frac{N}{R}\rfloor;N-R\lfloor\frac{N}{R}\rfloor+r \rangle$, thus the ITT between these two tokens $\tau$ is also \textbf{type 1}; 
otherwise, since $N-R\lfloor\frac{N}{R}\rfloor<R,r\le R$, the next tokens $1$ can be only assigned to a TB in slot $t^{'}+\lfloor\frac{N}{R}\rfloor+1$, in particular, $\langle t^{'}+\lfloor\frac{N}{R}\rfloor+1;N-R(\lfloor\frac{N}{R}\rfloor-1)+r \rangle$, and the duration of the ITT between these two tokens $\tau$ is $\lfloor\frac{N}{R}\rfloor+1$ time slots.
We refer to the ITTs with this duration as \textbf{type 2} ITTs.

Remark that all ITTs only have \textbf{type 1} and \textbf{type 2} for any $r$.
Therefore, we can determine $\hat{X}_i(k)$ by checking the positions of tokens $\tau$ in the RU domain of the $\tau$-SSNTA patterns from slot $t-1$ to slot $t+\frac{\operatorname{lcm}(N,R)}{R}-2$.
Specifically, the positions of tokens $\tau$ in order follow a circular shift process with the starting point at an RU and period $N-R\lfloor\frac{N}{R}\rfloor$ in ascending direction on the RU domain, and the last shift is back to the RU for the first time.
Every time the token $\tau$ position shifts through RU $R$, a \textbf{type 2} ITT will occur, and a \textbf{type 1} ITT will occur otherwise.
Based on the analysis above, the number of times the token $\tau$ position completes a circular shift is the number of \textbf{type 2} ITTs in the CTF, which can be calculated by $\frac{\operatorname{lcm}(N-R\lfloor\frac{N}{R}\rfloor,R)}{R}$.

This completes the proof.

\section{Proof of Theorem \ref{TheoremAAoI}}\label{Proof_AAoI}
First, we deduce that under the critical assumption, $\mathbb{E}[\hat{D}_i(k)]$ can be expressed by
\begin{equation}\label{EDk}
    \begin{split}        \mathbb{E}\left[\hat{D}_i(k)\right]&=\overline{S}_i\sum^k_{z=1}\hat{p}_{i,z}\prod^k_{j=z+1}(1-\hat{p}_{i,j})+\overline{\Delta}_i\prod^k_{j=1}(1-\hat{p}_{i,j})\\
    &+\sum^{k-1}_{z=0}\hat{X}(z)\prod^k_{j=z+1}(1-\hat{p}_{i,j})
    \end{split}
\end{equation}
by the induction method, where we specify $\prod^k_{j:j>k}(1-\hat{p}_{i,j})=1$ and $\hat{X}_i(0)=1$.
In particular, due to the assumption that $\Delta_i(t)$ and $\delta_i(t)$ are in steady state before the CTF considered, we can present
\begin{equation}
    \begin{split}
        \mathbb{E}\left[\hat{D}(1)\right]&=\hat{p}_{i,1}\sum_{\Delta,\delta}\theta_i(\Delta,\delta)(\delta+1)\\
        &+(1-\hat{p}_{i,1})\sum_{\Delta,\delta}\theta_i(\Delta,\delta)(\Delta+1)\\
        &=\hat{p}_{i,1}\overline{S}_i+(1-\hat{p}_{i,1})(\overline{\Delta}_i+1),
    \end{split}
\end{equation}
where $\theta_i(\Delta,\delta)$ denotes the steady-state joint distribution of AoI and local age of sensor $i$.
We assume the expression of $\mathbb{E}[\hat{D}_i(k)]$ follows \eqref{EDk}.
Then, 
\begin{equation}
    \begin{split}
        &\mathbb{E}\left[\hat{D}_i(k+1)\right]\\
        &=\hat{p}_{i,k+}\sum_{\Delta,\delta}\theta_i(\Delta,\delta)(\delta+1)\\
        &+(1-\hat{p}_{i,k+1})\left(\mathbb{E}\left[\hat{D}_i(k)\right]+\hat{X}_i(k)\right)\\
        &=\overline{S}_i[\hat{p}_{i,k+1}+(1-\hat{p}_{i,k+1})\sum^k_{z=1}\hat{p}_{i,z}\prod^k_{j=z+1}(1-\hat{p}_{i,j})]\\
        &+\overline{\Delta}_i\prod^{k+1}_{j=1}(1-\hat{p}_{i,j})+\sum^{k-1}_{z=0}\hat{X}(z)\prod^{k+1}_{j=z+1}(1-\hat{p}_{i,j})\\
        &+(1-\hat{p}_{i,k+1})\hat{X}_i(k)\\
        &=\overline{S}_i\sum^{k+1}_{z=1}\hat{p}_{i,z}\prod^{k+1}_{j=z+1}(1-\hat{p}_{i,j})+\overline{\Delta}_i\prod^{k+1}_{j=1}(1-\hat{p}_{i,j})\\
        &+\sum^{k}_{z=0}\hat{X}(z)\prod^{k+1}_{j=z+1}(1-\hat{p}_{i,j}).
    \end{split}
\end{equation}
Clearly, \eqref{EDk} still holds when we consider $k+1$.
This completes the derivation of $\mathbb{E}[\hat{D}_i(k)]$.

Substituting \eqref{QEC} and \eqref{EDk} into \eqref{AoICTF}, leveraging
\begin{equation}
    \sum^k_{z=1}\hat{p}_{i,z}\prod^k_{j=z+1}(1-\hat{p}_{i,j})=1-\prod^k_{j=1}(1-\hat{p}_j),
\end{equation}
and after some manipulations, we can obtain \eqref{AAoIApprox}.

\section{Proof of Theorem \ref{TheoremMod}}\label{ProofTheoremMod}

Consider the token assignment in the first $\frac{\operatorname{lcm}(N,R)}{R}$ time slots.
Suppose token $\tau$ is assigned to the $x$th TB, which resides in RU $a$.
The D-TFDA scheme periodically assigns tokens $1,2,\cdots,N$ to $\operatorname{lcm}(N,R)$ TBs following the order of TBs.
Under this construction, token $\tau$ and the $x$th TB satisfy the mapping that $\tau-1$ is obtained from $x-1$ by taking modulo by $N$, i.e., 
\begin{equation}
    (x-1) \equiv (\tau-1) \pmod{N},
\end{equation}
which implies $x \equiv \tau \pmod{N}$.
Meanwhile, since the token assignment array has $R$ rows corresponding to the $R$ RUs, we also have $x \equiv a \pmod{R}$.
By the generalized Chinese remainder theorem \cite{ore1952general}, $a$ and $\tau$ must satisfy 
\begin{equation}
    a \equiv \tau \pmod{\gcd(N,R)}.
\end{equation}
Furthermore, by Lemma~\ref{Period}, the token assignment pattern repeats every $\frac{\operatorname{lcm}(N,R)}{R}$ slots in the time dimension.
Hence, the token assignment in any block of $\frac{\operatorname{lcm}(N,R)}{R}$ consecutive time slots is a circular-shifted version of that in the first $\frac{\operatorname{lcm}(N,R)}{R}$ slots in the time dimension, while the RU positions in which token $\tau$ appears remain invariant across all time slots.

This completes the proof.

\section{Proof of Corollary \ref{CorollaryRU}}\label{ProofCorollaryRU}
Recall that token $\tau$ is assigned exactly $K = \frac{\operatorname{lcm}(N,R)}{N}$ times in any consecutive $\frac{\operatorname{lcm}(N,R)}{R}$ time slots.
Consider the token assignment in the first $\frac{\operatorname{lcm}(N,R)}{R}$ slots.
By Theorem~\ref{TheoremMod}, the RU $a$ in which token $\tau$ appears satisfies $a \equiv \tau \pmod{\gcd(N,R)}$.
This indicates that the collection of indices of RUs in which token $\tau$ appears is given by
\begin{equation}
    \left\{a \mid a = \tau + \alpha\gcd(N,R),\ \alpha \in \mathbb{Z}\right\} \cap \{1, 2, \cdots, R\},
\end{equation}
and the gap between any two adjacent RUs in which token $\tau$ appears must be $\gcd(N,R)$.
Let $\tau = \alpha^{\prime}\gcd(N,R) + \beta$, where $\alpha^{\prime}, \beta \in \mathbb{Z}$ and $\beta < \gcd(N,R)$.
Then, by the above set, all possible values of $\alpha$ are
\begin{equation}
    \left\{-\alpha^{\prime},\ 1-\alpha^{\prime},\ 2-\alpha^{\prime},\ \cdots,\ \frac{R}{\gcd(N,R)}-1-\alpha^{\prime}\right\}.
\end{equation}
Thus, token $\tau$ appears in exactly $\frac{R}{\gcd(N,R)}$ distinct RUs.
Furthermore, we have
\begin{equation}
    \frac{R}{\gcd(N,R)} = \frac{R}{\frac{NR}{\operatorname{lcm}(N,R)}} = \frac{\operatorname{lcm}(N,R)}{N} = K,
\end{equation}
yielding that token $\tau$ locates in $K$ different RUs.
Finally, by the proof of Lemma~\ref{Period}, the token assignment over any consecutive $\frac{\operatorname{lcm}(N,R)}{R}$ slots is a circular-shifted version of that in the first $\frac{\operatorname{lcm}(N,R)}{R}$ slots in the time dimension, where the RU locations of token $\tau$ remain unchanged.

The result follows.

\section{Proof of Proposition \ref{PropositionAETC}}\label{ProofPropositionAETC}
Recall that the $N$ tokens are periodically assigned to the TBs following the order of the TBs, and there are $R$ RUs in each time slot.
Based on this, for a token $\tau$ that appears in RU $r$ in the current time slot, the RU $r^{\prime}$ in which this token appears in the next time slot is given by
\begin{equation}
    r^{\prime} = \big((r-1 + N) \bmod R\big) + 1.
\end{equation}
Accordingly, the RU-location change of a token depends only on the RU that the token currently occupies.
As such, once several tokens appear together in the same RU within an assignment of the consecutive token sequence $1, \cdots, N$, they will always appear together in any subsequent occurrence of that sequence.
Furthermore, recall that the time-slot interval between two adjacent appearances of token $\tau$ depends solely on the RU in which the token resides.
That is, the single token assignment patterns across time slots of tokens that can appear in the same RU are equivalent in the long term.
Consequently, the AAoI performances of a sensor assigned with any of these tokens are identical.

Consider tokens $1, \cdots, N$ as assigned in the first $N$ TBs.
By Theorem~\ref{TheoremMod}, Corollary~\ref{CorollaryRU}, and the evolution of AoI, the set of tokens that appear in all $\gcd(N,R)$-interval $K$ RUs, where any token $\tau$ can appear, follows the definition of an AETC. 
Clearly, the AETCs containing tokens $\tau$ and $\tau^{\prime}$ are the same if these tokens can appear in any identical RU.
We now determine the number of distinct AETCs.
The $K$ RUs in which token $1$ can appear are RUs
\begin{equation}
 1,\ 1 + \gcd(N,R),\ \cdots,\ 1 + (K-1)\gcd(N,R),
\end{equation}
where $(K-1)\gcd(N,R) = R - \gcd(N,R)$.
Notice that the $K$ RUs in which any other token can appear are circular-shifted versions of those of token $1$ in the RU dimension.
Hence, it is straightforward that the AETCs containing tokens $1, \cdots, \gcd(N,R)$ are distinct clusters, and the AETC containing any other token coincides with one of these $\gcd(N,R)$ clusters.
Consequently, there are exactly $\gcd(N,R)$ distinct AETCs in total.

For the tokens $1, \cdots, N$, each of the first $N - \left\lfloor \frac{N}{R} \right\rfloor R$ RUs contains $\left\lfloor \frac{N}{R} \right\rfloor + 1$ tokens, and each of the remaining RUs contains $\left\lfloor \frac{N}{R} \right\rfloor$ tokens.$\left\lfloor \frac{N}{R} \right\rfloor R$
For token $\tau \in \{1, \cdots, \gcd(N,R)\}$, the first $\frac{N - \left\lfloor \frac{N}{R} \right\rfloor R}{\gcd(N,R)}$ RUs in which it can appear are RUs
\begin{equation}
    \tau,\ \tau + \gcd(N,R),\ \cdots,\ \tau + \left(\frac{N - \left\lfloor \frac{N}{R} \right\rfloor R}{\gcd(N,R)} - 1\right)\gcd(N,R).
\end{equation}
This indicates that for all $\gcd(N,R)$ AETCs, the first $\frac{N - \left\lfloor \frac{N}{R} \right\rfloor R}{\gcd(N,R)}$ associated RUs belong to the set $\{1, \cdots, N - \left\lfloor \frac{N}{R} \right\rfloor R\}$, while the remaining $K - \frac{N - \left\lfloor \frac{N}{R} \right\rfloor R}{\gcd(N,R)}$ associated RUs belong to $\{N - \left\lfloor \frac{N}{R} \right\rfloor R + 1, \cdots, R\}$.
Therefore, each AETC contains
\begin{equation}
    \begin{aligned}
        &\frac{N - \left\lfloor \frac{N}{R} \right\rfloor R}{\gcd(N,R)}\left(\left\lfloor \frac{N}{R} \right\rfloor + 1\right) + \left(K - \frac{N - \left\lfloor \frac{N}{R} \right\rfloor R}{\gcd(N,R)}\right)\left\lfloor \frac{N}{R} \right\rfloor\\
        &= K\left\lfloor \frac{N}{R} \right\rfloor + \frac{N - \left\lfloor \frac{N}{R} \right\rfloor R}{\gcd(N,R)}
        = \frac{N}{\gcd(N,R)}=\Gamma
    \end{aligned}
\end{equation}
distinct tokens.
It is worth noting that in the network token assignment pattern, the consecutive tokens $1,2,\ldots,N$ undergo repeated cyclic shifts in the RU domain with a period of $N-\left\lfloor \frac{N}{R} \right\rfloor R$.
Since $N-\left\lfloor \frac{N}{R} \right\rfloor R$ is divisible by $\gcd(N,R)$, the partitions of AETCs in different circular-shifted versions of the consecutive $N$ tokens are the same.

Then the proposition follows.

\end{document}